\documentclass[11pt,letterpaper]{article}

\usepackage[compact]{titlesec}
\usepackage{amsmath,amsthm,nicefrac}
\usepackage{amsfonts, amstext}
\usepackage{mathtools}

\usepackage{typearea}
\typearea{14}
\usepackage[font={small,it}]{caption}

\usepackage{setspace}

\usepackage[shortlabels]{enumitem}
\usepackage[breaklinks]{hyperref}
\hypersetup{colorlinks=true,%
            citebordercolor={.6 .6 .6},linkbordercolor={.6 .6 .6},%
citecolor=blue,urlcolor=black,linkcolor=blue,pagecolor=black}

\usepackage[nameinlink]{cleveref}
\Crefname{algocf}{Algorithm}{Algorithms}
\crefname{algocfline}{line}{lines}
\Crefname{invariant}{Invariant}{Invariants}
\Crefname{claim}{Claim}{Claims}
\Crefname{subclaim}{Subclaim}{Subclaims}

\usepackage{epsfig}
\usepackage{amsthm,amssymb}
\usepackage{amsthm,amssymb}

\usepackage{mathrsfs}
\usepackage{xspace}
\usepackage{soul}
\usepackage{latexsym}
\usepackage{bbm}

\usepackage{framed}

\usepackage[dvipsnames]{xcolor}
\definecolor{DarkGray}{rgb}{0.66, 0.66, 0.66}
\definecolor{DarkPowderBlue}{rgb}{0.0, 0.2, 0.6}
\definecolor{fluorescentyellow}{rgb}{0.8, 1.0, 0.0}
\definecolor{cerulean}{rgb}{0.0, 0.48, 0.65}
\definecolor{bleudefrance}{rgb}{0.19, 0.55, 0.91}

\usepackage[ruled,vlined,linesnumbered,algonl]{algorithm2e}
\SetEndCharOfAlgoLine{}
\SetKwComment{Comment}{\footnotesize$\triangleright$\ }{}

\SetCommentSty{mycommfont}

\usepackage{thmtools,thm-restate}

\allowdisplaybreaks

\newcounter{note}[section]

\sethlcolor{fluorescentyellow}

\newcommand{\initOneLiners}{%
    \setlength{\itemsep}{0pt}
    \setlength{\parsep }{0pt}
    \setlength{\topsep }{0pt}
}

\newtheorem{theorem}{Theorem}[section]
\newtheorem{lemma}[theorem]{Lemma}

\newtheorem{fact}[theorem]{Fact}
\newtheorem{corollary}[theorem]{Corollary}

\theoremstyle{definition}
\newtheorem{defn}[theorem]{Definition}

\theoremstyle{remark}

\makeatletter 
\renewcommand{\theinvariant}{(I\@arabic\c@invariant)}
\makeatother

\newcommand{\tO}{\tilde{O}}

\newcommand{\var}{{\mathbb{V}}}
\newcommand{\ex}{{\mathbb{E}}}
\newcommand{\E}{\mathbb{E}}
\newcommand{\relv}{\eta}
\newcommand{\ugp}{u_G(p)}
\newcommand{\zgp}{z_G(p)}
\newcommand{\xgp}{x_G(p)}
\newcommand{\Cresp}{{\cal C}_7}
\newcommand{\zresp}{z_7(p)}

\newcommand{\euler}{{\cal E}}
\newcommand{\odd}{{\rm odd}}
\newcommand{\even}{{\rm even}}
\newcommand{\bS}{\mathbb{S}}
\newcommand{\bL}{\mathbb{L}}
\newcommand{\diff}{{\rm d}}

\newcommand{\tX}{\tilde{X}}

\newcommand{\eps}{\varepsilon}

\newcommand{\EE}{\mathbb{E}}
\newcommand{\RR}{\mathbb{R}}

\newcommand{\poly}{\operatorname{poly}}

\renewcommand{\emptyset}{\varnothing}

\newcommand{\tG}{{\tilde G}}

\newcommand{\nf}{\nicefrac}

\newcommand{\junk}[1]{}
\newcommand{\eat}[1]{}

\newif\ifhideproofs

\ifhideproofs
\usepackage{environ}
\NewEnviron{hide}{}

\fi

\title{Beyond the Quadratic Time Barrier for Network Unreliability}
\author{Ruoxu Cen\thanks{Department of Computer Science, Duke University. Email: {\tt ruoxu.cen@duke.edu}} \and William He\thanks{Duke University. Email: {\tt william.rui.he@gmail.com}} \and Jason Li\thanks{Simons Institute for the Theory of Computing, UC Berkeley. Email: {\tt jmli@alumni.cmu.edu}} \and Debmalya Panigrahi\thanks{Department of Computer Science, Duke University. Email: {\tt debmalya@cs.duke.edu}}}
\date{}

\begin{document}

\pagenumbering{gobble}

\maketitle

\begin{abstract}
Karger (STOC 1995) gave the first FPTAS for the network (un)reliability problem, setting in motion research over the next three decades that obtained increasingly faster running times, eventually leading to a $\tilde{O}(n^2)$-time algorithm (Karger, STOC 2020). This represented a natural culmination of this line of work because the algorithmic techniques used can enumerate $\Theta(n^2)$ (near)-minimum cuts. In this paper, we go beyond this quadratic barrier and obtain a faster FPTAS for the network unreliability problem. Our algorithm runs in $m^{1+o(1)} + \tilde{O}(n^{1.5})$ time. 

Our main contribution is a new estimator for network unreliability in \emph{very} reliable graphs. These graphs are usually the bottleneck for network unreliability since the disconnection event is elusive. Our estimator is obtained by defining an appropriate importance sampling subroutine on a dual spanning tree packing of the graph. To complement this estimator for very reliable graphs, we use recursive contraction for \emph{moderately} reliable graphs. We show that an interleaving of sparsification and contraction can be used to obtain a better parametrization of the recursive contraction algorithm that yields a faster running time matching the one obtained for the very reliable case.
\end{abstract}

\clearpage

\pagenumbering{arabic}

\section{Introduction}
\label{sec:introduction}
In the network unreliability problem, we are given an undirected, unweighted graph $G=(V, E)$ and a failure probability $0 < p < 1$. The goal is to find the probability that the graph disconnects when every edge is deleted independently with probability $p$. The probability of disconnection is called the unreliability of the graph and denoted $\ugp$. 
Reliability problems naturally arise from the need to keep real-world networks connected under edge failures, and entire books have been devoted to the topic~\cite{colbourn1987combinatorics, chaturvedi2016network}. In particular, the problem of estimating $\ugp$ has been dubbed as ``\ldots perhaps the most fundamental problem in the study of network reliability''\cite{Karger20}.

In 1979, Valiant showed that the network unreliability problem is $\mathsf{\#P}$-hard~\cite{Valiant79}, which implies that it is unlikely that a polynomial-time algorithm can {\em exactly} compute the value of $u_G(p)$. Over the next 15 years, several algorithms were designed to approximate $\ugp$ for various special cases such as planar graphs (Karp and Luby~\cite{karp1989monte}), dense graphs (Alon, Frieze, and Welsh~\cite{AlonFW95}), etc. Eventually, in a seminal work, Karger \cite{Karger99} proposed the first fully polynomial-time randomized approximation scheme (FPRAS) for the unreliability problem. For any constant $\eps \in (0, 1)$, this algorithm outputs a $(1\pm\eps)$-approximation for $\ugp$ with high probability (whp)\footnote{In this paper, as in the network unreliability literature, a result is said to hold with high probability if it holds with probability $1-\frac{1}{\poly(n)}$.} in $\tO(mn^4)$ time, where $n$ is the number of vertices in the graph and $m$ is the number of edges. This work established a bifurcation of instances of the problem into two cases depending on the value of $\ugp$. If $\ugp$ is {\em large} (the unreliable case), then Monte Carlo sampling suffices. On the other hand, if $\ugp$ is {\em small} (the reliable case), then Karger showed that whenever the graph disconnects, it's almost always the case that a {\em near-minimum} cut fails. In this case, the algorithm uses DNF counting~\cite{karp1989monte} on the (polynomial) set of near-minimum cuts to obtain an estimate of $\ugp$. Using the same template, Harris and Srinivasan \cite{HarrisS18} improved the running time of the algorithm to $n^{3+o(1)}$. They did so by establishing a tighter bound on the number of cuts that contribute to $u_G(p)$ and showing that the instances of DNF counting generated by the unreliability problem are nongeneric and admit faster algorithms than the generic case.

In the last decade, Karger \cite{Karger16, Karger17, Karger20} has further improved the running time of the problem to the current best bound of $\tO(n^2)$. As in prior works, the unreliable case is handled by Monte Carlo sampling. But, for the reliable case, instead of DNF counting, every edge is failed with some probability $q\geq p$ and all edges that survive are contracted to yield a smaller graph $H$. On this smaller graph $H$, the algorithm computes $u_H(p/q)$, which is an unbiased estimator of $\ugp$. Now, if the relative variance of $u_H(p/q)$ is bounded by $\eta$, then a standard technique of repeating $\tO(\eta\eps^{-2})$ times and taking the {\em median-of-averages} yields a $(1\pm\eps)$-approximation to $\ugp$ whp. 

The algorithms in \cite{Karger16, Karger17, Karger20} differ in the choice of $q$, the structure of the computation tree, and the bound on relative variance.\footnote{The relative variance of a random variable is the ratio of its variance to the square of its expectation.} In \cite{Karger16}, $q$ is chosen so that $H$ is expected to be of constant size, and the resulting estimator has a relative variance of $O(n^2)$ and can be computed in $\tO(n)$ time. This yields a running time of $\tO(n^3)$. In \cite{Karger17}, $q$ is chosen differently, so as to ensure that the expected size of $H$ is a constant fraction of $G$. The algorithm then proceeds to recursively estimate $u_H(p/q)$, closely resembling the recursive contraction framework of Karger and Stein~\cite{karger1996new} originally proposed for the minimum cut problem. The bound on the running time of this algorithm was initially established as $\tO(n^{2.71})$~\cite{Karger17}, but a more refined analysis of the relative variance later improved this to $\tO(n^2)$~\cite{Karger20}. This remained the fastest algorithm for the network unreliability problem prior to our work.

\subsection{Our Contribution}

In this paper, we give a fully polynomial-time randomized approximation scheme (FPRAS) for the network unreliability problem that has a running time of $m^{1+o(1)} + \tO(n^{1.5})$. We show the following theorem:

\begin{theorem}\label{thm:main}
    For any \eat{constant} $\eps\in (0, 1)$, there is a randomized Monte Carlo algorithm for the network unreliability problem that runs in $m^{1+o(1)}\eps^{-2} + \tO(n^{1.5}\eps^{-3})$ time and outputs a $(1\pm\eps)$-approximation to $\ugp$ whp.
\end{theorem}

We remark that the dependence on $m$, which is not explicit in prior bounds because the bounds are $\Omega(n^2)$ and therefore $\Omega(m)$, is necessary because the number of edges may exceed $n^{1.5}$ in general. 


The running time bound of $\tO(n^2)$ obtained by \cite{Karger20} was a natural culmination of prior research on the unreliability problem because all prior algorithms involve explicit or implicit enumeration of a set of $\Omega(n^2)$ (near-)minimum cuts of the graph. To see this, consider the cycle graph on $n$ vertices. This is a sparse graph of $n$ edges, and one would ideally like a subquadratic running time for the unreliability problem. However, the cycle has ${n\choose 2} = \Theta(n^2)$ minimum cuts, which means that any subquadratic algorithm must refrain from enumerating the minimum cuts of the graph. The prior techniques for the unreliability problem are unsuitable for this purpose. The DNF counting-based algorithms~\cite{Karger99sparsify,HarrisS18} perform explicit enumeration of near-minimum cuts and the DNF formula is already of $\Omega(n^2)$ size. The later algorithms based on recursive contraction~\cite{Karger16,Karger17,Karger20} do not perform explicit enumeration, but the framework can be used to generate all the (possibly $\Theta(n^2)$) minimum cuts of a graph. Our main conceptual contribution in this work is to overcome this bottleneck of $\Omega(n^2)$ in the running time of unreliability algorithms. We outline the main techniques that we employ for this purpose next.


\subsection{Our Techniques}\label{sec:techniques}



As in prior work, in the {\em unreliable} case, i.e., when $u_G(p)$ is above some threshold $n^{-o(1)}$, we apply na\"ive Monte Carlo sampling. The algorithm removes each edge independently with probability $p$ and checks whether the remaining graph is connected. Each sample takes $O(m)$ time and the estimator has relative variance $\approx \frac{1}{\ugp}$. Thus, $\tO\left(\frac{1}{u_G(p)}\right)=n^{o(1)}$ samples are sufficient and the running time of the algorithm is $m^{1+o(1)}$. 

When $\ugp$ is too small for Monte Carlo sampling, Karger~\cite{Karger20} proved that the partition function $\zgp$, defined as the expected number of failed cuts, approximates $\ugp$ up to a factor of $1+o(1)$.  So, it suffices to estimate $\zgp$ in this case instead of $\ugp$. 

Let us first review the recursive contraction algorithm in \cite{Karger20} for estimating $\zgp$. The algorithm removes edges with probability $q > p$, contracts the edges that remain to form a smaller graph $H$ (putting back the edges removed initially), and estimates $u_H(p/q)$ in this graph recursively. The parameter $q$ is chosen so that $q^{-\lambda} = 2$ (or any constant), where $\lambda$ is the value of the minimum cut in the graph. This ensures that one step of contraction reduces the number of vertices by a constant fraction in expectation. The relative variance is $q^{-\lambda}(1+o(1))$. By repeating $q^{-\lambda}$ times, the relative variance is kept at $1+o(1)$. This yields the recurrence:
$T(n)=q^{-\lambda}\cdot T(q^{\lambda/2} \cdot n)$.
Solving this recurrence gives $T(n)=O(n^2)$.

In this recurrence, the size bound and relative variance can both be simultaneously tight. To see this, consider a cycle graph with $\frac{\lambda}{2}+1$ parallel edges where one vertex on the cycle is connected to a leaf vertex with $\lambda$ parallel edges. Suppose $p$ is small enough that the minimum cut dominates unreliability, i.e., $\zgp \approx p^\lambda$. The estimator (approximately) returns $(p/q)^\lambda$ with probability $q^\lambda$; hence, the relative variance is $\frac{1}{\zgp^2}\cdot q^\lambda \cdot \left(\frac{p}{q}\right)^{2\lambda} \approx q^{-\lambda}$. The expected size of the contracted graph is $\approx q^{\lambda/2+1}\cdot n$. This example shows that we cannot hope to unconditionally improve the previous analysis to obtain a sub-quadratic bound.

But, under what conditions can the pathological example above be realized? To answer this question, we partition the reliable case into two subcases depending on whether $\ugp$ exceeds $O(n^{-3})$. We call these the {\em moderately reliable} and {\em very reliable} cases respectively. Clearly, the example above can be realized in the very reliable case; hence, we need a new algorithm for this case. This algorithm is our main contribution, and we will describe it below. But, before doing so, let us consider the moderately reliable setting. In this case, our main observation is that the two bounds on graph size and relative variance in the recurrence relation for recursive contraction cannot be tight simultaneously. Consider the following two typical cases. In a cycle with $\frac{\lambda}{2}$ parallel edges, the size bound $q^{\lambda/2}\cdot n$ is tight, but the true relative variance is close to 1, which is much smaller than the $q^{-\lambda}$ bound. In contrast, in a graph with a single minimum cut where all other cuts are large (e.g., a dumbbell graph), the relative variance is close to $q^{-\lambda}$ but the size decreases much faster than $q^{\lambda/2}\cdot n$. In general, we interpolate between these two extreme cases. We choose a parameter $\gamma \in \left[\frac{3}{4}\lambda, \lambda\right]$, and prove a sharper bound of $q^{-\gamma }$ on relative variance and $q^{\frac{2\gamma}{3}}\cdot n$ on the size of the recursive instance.
Then the recurrence becomes $T(n) = q^{-\gamma }\cdot T\left(q^{\frac{2\gamma}{3}}\cdot n\right)$, which yields a running time of $O(n^{1.5})$ as desired.

Finally, we consider the very reliable case, i.e., when $\ugp$ is smaller than $O(n^{-3})$. In this case, we design a new algorithm for estimating $\zgp$. When $p$ is small, the contribution of large cuts to the value of $\zgp$ decreases. In particular, for $\ugp = O(n^{-3})$, we can discard all cuts of value $> 3.5\lambda$. Our goal, then, is to obtain an estimator that performs {\em importance sampling} on cuts of value $\le 3.5\lambda$ (i.e., reweights them appropriately) to yield $\zgp$. But, how do we sample such cuts? We use a semi-duality between (near-)minimum cuts and maximum spanning tree packings for this purpose. Namely, one can construct $\lambda$ spanning trees with edge congestion $2$, which implies that in a randomly chosen tree, a cut of value $\le \alpha\lambda$ only has $\le 2\alpha$ edges, i.e., {\em $2\alpha$-respects} the tree, in expectation. (See Gabow~\cite{Gabow95} and Karger~\cite{Karger00mincut} for prior uses of such duality in minimum cut algorithms.) This allows us to sample $7$-respecting cuts from a spanning tree instead of cuts of value $\le 3.5\lambda$, and redefine the support of the estimator to $7$-respecting cuts. There are two main challenges: first, we need to implement the importance sampling subroutine very efficiently, namely calculate each sample and its corresponding weight in $\tO(1)$ time, and second, we must control the relative variance to $\tO(n^{1.5})$ since we can only draw $\tO(n^{1.5})$ samples (note that the support of the distribution is $\approx n^7$). We design a data structure based on orthogonal range queries in $\RR^2$ to implement importance sampling efficiently. To control the relative variance, we contract well-connected components of the graph (using a Gomory-Hu tree~\cite{gomory1961multi}) and apply the sampling subroutine on this contracted graph to boost the sampling probability of the small cuts. 

\eat{
During random contraction, the edge size will be bounded. By Lemma \ref{lem:contraction-size-bound}, the expected edge size of $G(q)$ is at most $n/(1-q)= O(n\lambda)$. So the recurrence of balanced case is
\[T(n,m) = m^{1+o(1)} + \tO(n^{1.5}) + q^{-\gamma }T(q^{\frac 23\gamma }n, O(n\lambda))\]
The solution is $T(n,m) = \tO(m^{1+o(1)}+n^{1.5}\lambda^{1+o(1)})$. Finally, we apply a sparsification step to reduce $\lambda$ to $\tO(1)$, so that the running time is $m^{1+o(1)}+\tO(n^{1.5})$.
}

\subsection{Overall Algorithm and Paper Organization} 


We now give an overview of the algorithm. We describe when the algorithm invokes its different components, and give a pointer to the section where each component is described and analyzed. 

Our overall algorithm builds a recursive computation tree. There are several base cases of this recursion that we describe first. 
\begin{itemize}
    \item The first base case is determined simply the number of vertices in the graph. If $n \le \tO(\eps^{-2})$, then we run Karger's algorithm \cite{Karger20} that gives an unbiased estimator for $\ugp$. The following theorem states the properties of this estimator: 
    \begin{theorem}[\cite{Karger20}]\label{thm:karger-estimator}
        Given a graph $G$ with vertex size $n$, an unbiased estimator of $\ugp$ with $O(1)$ relative variance can be computed in $\tO(n^2)$ time. As a consequence, a $(1\pm\eps)$-approximation to $\ugp$ can be computed in time $\tO(n^2\eps^{-2})$, which is $\tO(n^{1.5}\eps^{-3})$ for $n = \tO(\eps^{-2})$. 
    \end{theorem} 

    \item The second base case is to run Monte Carlo sampling. Intuitively, this is done when the probability of the graph being disconnected is large. There are two subcases for Monte Carlo sampling: 
        \begin{itemize} 
            \item The first subcase is when $p > \theta$ for some threshold $\theta$ whose value will be given by \Cref{lem:z-approx-u}. In this case, we run a na\"ive Monte Carlo algorithm to obtain an estimator for $\ugp$ (\Cref{sec:naive}). The properties of the estimator are summarized below:
            
            \begin{restatable}{lemma}{mc}\label{lem:mc}
                For any $p\ge\theta$, an unbiased estimator of $\ugp$ with relative variance $O(1)$ can be computed in time $m^{1+o(1)}$. As a consequence, a $(1\pm\eps)$-approximation to $\ugp$ can be computed in $m^{1+o(1)}\eps^{-2}$ time under the condition that $p\ge\theta$.
            \end{restatable}
            
            \item The second subcase is when $p < \theta$ but $p^{\lambda} > n^{-0.5}$. In this case, we run a {\em two-step} Monte Carlo algorithm to obtain an estimator for $\ugp$ (\Cref{sec:two-step}). 
            The properties of the estimator are summarized below:

            \begin{restatable}{lemma}{mctwo}\label{lem:mc-2step}
                For $p$ such that $p<\theta$ and $p^\lambda > n^{-1/2}$, an unbiased estimator of $\ugp$ with relative variance $O(1)$ can be computed in $\tO(m+n^{1.5})$ time. As a consequence, a $(1\pm\eps)$-approximation to $\ugp$ can be computed in $\tO((m+n^{1.5})\eps^{-2})$ time under the condition that $p<\theta$ and $p^\lambda > n^{-1/2}$.
            \end{restatable}
        \end{itemize}
    
    \item The final base case is the most interesting new contribution of this paper. This is invoked in the highly reliable setting, when $p^\lambda < 4n^{-3}$. In this case, we run an importance sampling algorithm on a spanning tree packing of the graph (\Cref{sec:importance}).
    We prove the following lemma (an estimator $X$ for $\ugp$ with relative bias $\delta$ satisfies $\E[X] \in (1\pm \delta)\ugp$):
    \begin{restatable}{lemma}{importance}\label{lem:interface-importance}
        For any $p$ such that $p^\lambda \le O(n^{-3})$, an estimator for $\ugp$ with relative bias $O\left(\frac{\log n}{\sqrt{n}}\right)$ and relative variance $O(1)$ can be computed in $m^{1+o(1)}+\tO(n^{1.5})$ time. As a consequence, a $(1\pm\eps)$-approximation to $\ugp$ can be computed in $(m^{1+o(1)}+\tO(n^{1.5}))\eps^{-2}$ time under the condition that $p^\lambda \le O(n^{-3})$.
    \end{restatable}
\end{itemize}

We have described the base cases, all of which are non-recursive algorithms. The remaining case is when $4n^{-3} \le p^{\lambda} \le n^{-0.5}$ and $p < \theta$. In this case, we run a step of recursive contraction (\Cref{sec:contraction}). In earlier works using recursive contraction, the number of edges is trivially bounded by $O(n^2)$ in recursive calls. Since we would like to impose the stricter bound of $\tO(n^{1.5})$ on the running time, we need the number of edges to also satisfy this bound. But, this may not hold in general in recursive steps. To restore this property, we occasionally interleave calls to a standard sparsification algorithm with the recursive contraction steps. This increases variance -- we bound it in \Cref{sec:sparsify} and use this bound in the analysis of recursive contraction in \Cref{sec:contraction}. We obtain the following lemma (the relative second moment of a random variable $X$ is defined as $\EE[X^2]/(\EE[X])^2$):
\begin{restatable}{lemma}{contract}\label{lem:interface-contract}
    Suppose $4n^{-3} \le p^{\lambda} \le n^{-0.5}$ and $p < \theta$. An estimator $X$ for $\ugp$ with relative bias $\le 0.1\eps$ and relative second moment $\le \log^{O(1)} n$ can be computed in $m^{1+o(1)}+\tO(n^{1.5}\eps^{-1})$ time.
\end{restatable}

We now show that \Cref{thm:main} follows from this lemma:
\begin{proof}[Proof of \Cref{thm:main}]
The base cases are immediate from \Cref{thm:karger-estimator} and \Cref{lem:mc,lem:mc-2step,lem:interface-importance}. So, we focus on the recursive contraction case.
Let $X$ be the estimator output by the recursive contraction algorithm.
From \Cref{lem:interface-contract}, we have $\eta[X]\le \tO(1)$. 
By standard techniques (see \Cref{lem:mc-sample}), we can run the algorithm $\tO(\eps^{-2})$ times to get a $(1\pm\frac{\eps}{2})$-approximation of $\E[X]$ whp.
Because $\E[X]$ is a $(1\pm 0.1\eps)$-approximation of $\ugp$, the aggregated estimator is a $(1\pm \eps)$-approximation of $\ugp$.
Each run takes $m^{1+o(1)}+\tO(n^{1.5}\eps^{-1})$ time, so the overall running time is $m^{1+o(1)}\eps^{-2}+\tO(n^{1.5}\eps^{-3})$.
\end{proof}


Finally, we describe how the algorithm decides which case it is in at any node of the computation tree. The first base case can be identified based on the size of the graph. If not in the first base case, we calculate the value of the minimum cut $\lambda$ in $\tO(m)$ time~\cite{Karger00mincut}. If $p^{\lambda} < 4n^{-3}$, we are in the third base case. We are left to identify the two Monte Carlo base cases. We need the value of $\theta$ for this determination. Unfortunately, we do not know of a way to efficiently calculate $\theta$. Therefore, we distinguish identify these cases indirectly. We first run the na\"ive Monte Carlo algorithm and calculate the estimator of $\ugp$ given by this algorithm (which is the empirical probability of disconnection). If the value of this estimator is at least $n^{-o(1)}$, then we can conclude that the estimator $(1\pm\eps)$-approximates $\ugp$ whp. If the estimator returns a smaller value, then we are in the case $p < \theta$. In this case, we calculate the value of $p^\lambda$ and depending on whether it exceeds $n^{-0.5}$, we either run the two-step Monte Carlo algorithm or a step of recursive contraction.


\section{Preliminaries}
\label{sec:prelim}
We give some known observations in this section that we use throughout the paper.

\paragraph{Phase Transition.} An important observation due to Karger~\cite{Karger20} is the so called {\em phase transition} property of network unreliability. Roughly speaking, this property says that there is a threshold $\theta$ on the value of $p$ such that (a) above this threshold, a na\"ive Monte Carlo algorithm is efficient, and (b) below this threshold, conditioned on the graph getting disconnected, the probability that a single cut fails is close to 1. We state this property precisely below (\Cref{lem:z-approx-u}).

Let $\mathcal{C}(G)$ be the family of cuts in $G$, where a cut is represented by the set of cut edges.  Let $z_G(p)$ be the number of failed cuts, and $x_G(p)$ be the number of failed cut pairs.
By linearity of expectation,
\[z_G(p)=\sum_{C_i\in \mathcal{C}(G)}p^{|C_i|}
\quad \text{and} \quad
x_G(p)=\sum_{C_i, C_j\in \mathcal{C}(G), C_i\ne C_j} p^{|C_i\cup C_j|}\].
When context is clear, we write $z=z_G(p)$ and we omit the index range $\mathcal{C}(G)$ from the sums.

We state the following phase transition lemma: 
\begin{lemma}\label{lem:z-approx-u}
There exists a threshold $\theta$ such that
\begin{enumerate}
    \item $u_G(\theta) = n^{-O(1/\log \log n)}$.
    \item When $p<\theta$, we have $\frac{x_G(p)}{z_G(p)} \le \frac{1}{\log n}$; therefore, $\left(1-\frac{1}{\log n}\right) z_G(p) \le \ugp \le z_G(p)$.\footnote{In this paper, all logarithms are with base 2 unless otherwise mentioned.}
\end{enumerate}
\end{lemma}

In the very reliable case, i.e, $p < O(n^{-3})$, we can show a tighter bound for approximating $u$ with $z$:
\begin{lemma}\label{lem:z-approx-u-reliable}
    When $p^\lambda < O(n^{-3})$, we have $\frac{x_G(p)}{z_G(p)} \le O\left(\frac{\log n}{\sqrt{n}}\right)$; therefore, 
    \[
        \left(1-O\left(\frac{\log n}{\sqrt{n}}\right)\right) z_G(p) \le \ugp \le z_G(p).
    \]
\end{lemma}
The bounds in above two lemmas easily follow from previously known bounds in \cite{Karger20}. We give proof of the lemmas for completeness in \Cref{sec:proofs}. 

\paragraph{Random Contraction, Sparsification, and Near-Minimum Cuts.} Next, we give some standard results in graph connectivity related to random contractions, counting near-minimum cuts, and graph sparsification. First is a standard bound on the number of near-minimum cuts. Let $\lambda$ be the value of a minimum cut. Let the value of a cut $C_i$ be denoted $c_i$. We call a cut $C_i$ $d$-strong if $c_i\ge d$ and $d$-weak if $c_i\le d$.

\begin{lemma}[Lemma 3.2 of \cite{Karger00mincut}]\label{lem:cut-counting}
The number of $\alpha\lambda$-weak cuts in a graph with minimum cut value $\lambda$ is at most $n^{\lfloor 2\alpha\rfloor}$ for any $\alpha \ge 1$.
\end{lemma}

We also use the following standard sparsification result:
\begin{lemma}[Corollary 2.4 of \cite{Karger99sparsify}]\label{lem:sparsify}
Given an unweighted undirected graph $G$ with min-cut value $\lambda$ and any parameter $\delta\in(0,1)$, there exists $\alpha=O\left(\frac{\log n}{\delta^2\lambda}\right)$ such that if a subgraph $H$ is formed by picking each edge independently with probability $\alpha$ in $G$, then the following holds whp: for every cut $S$, its value in $H$ (denoted $d_H(S)$) and its value in $G$ (denoted $d_G(S)$) are related by $d_H(S)\in [(1-\delta)\cdot \alpha\cdot d_G(S), (1+\delta)\cdot \alpha\cdot d_G(S)]$. Note that this implies that the min-cut value in $H$ is $\lambda_H = O(\log n/\delta^2)$.

The running time of the sparsification algorithm is $O(m)$.
\end{lemma}

Finally, we state a bound on the expected number of uncontracted edges after random edge contractions. This was previously used (and proved) in the celebrated linear-time randomized MST algorithm of Karger, Klein, and Tarjan.
\begin{lemma}[Lemma 2.1 of \cite{KargerKT95}]\label{lem:contraction-size-bound}
Given an undirected multigraph, if we contract each edge independently with probability $\pi$, then the expected number of uncontracted edges is at most $n/\pi$. 
\end{lemma}

\paragraph{Gomory-Hu Tree.} Next, we recall the definition of a Gomory-Hu tree ~\cite{gomory1961multi}, which is used in various parts of the paper:
\begin{defn}\label{def:gomory-hu}
The {Gomory-Hu tree} \cite{gomory1961multi} of an undirected graph $G = (V, E)$ is a (weighted) tree $Y$ on the same set of vertices $V$ such that for every pair of vertices $s, t\in V$, the minimum $(s, t)$-cut in $Y$ (which is simply the minimum weight edge on the unique $s-t$ path in $Y$) is also a minimum $(s, t)$-cut in $G$, and has the same value.
\end{defn}

\paragraph{Relative Variance and Relative Bias.} The {\em relative variance} of our estimators will play an important role in our analysis. We define this below and state some standard properties. 
\begin{defn}
The {\em relative variance} of a random variable $X$, denoted $\relv[X]$, is defined as the ratio of its variance and its squared expectation, i.e., $\relv[X] = \frac{\var[X]}{\ex^2[X]} = \frac{\ex[X^2]}{\ex^2[X]}-1$. 
We also define {\em relative second moment} of $X$ as $\frac{\ex[X^2]}{\ex^2[X]} = \relv[X]+1$.
\end{defn}
Similar to variance, relative variance can be decreased by taking multiple independent samples.
\begin{fact}[Lemma I.4 of \cite{Karger17}]\label{fact:rel-var-decrease}
The average of $N$ independent samples of $X$ has relative variance $\frac{\eta[X]}{N}$.
\end{fact}
This leads to the following property:
\begin{lemma}[Lemma I.2 of \cite{Karger17}]\label{lem:mc-sample}
Fix any $\eps, \delta \in (0, 1)$. For a random variable $X$ with relative variance $\relv[X]$,  the median of $O\left(\log\frac{1}{\delta}\right)$ averages of $O\left(\frac{\relv[X]}{\eps^2}\right)$ independent samples of $X$ is a $(1\pm \eps)$-approximate estimation of $\E[X]$ with probability $1-\delta$.
\end{lemma}

The next lemma is an important property of relative variance that allows us to compose estimators in a recursive algorithm.
\begin{lemma}[Lemma II.3 of \cite{Karger17}]\label{lem:relvar-multiply}
Suppose $Y$ is an unbiased estimator of $x$ with relative variance $\eta_1$, and conditioned on a fixed $Y$, $Z$ is an unbiased estimator of $Y$ with relative variance $\eta_2$. Then $Z$ is an unbiased estimator for $x$ with relative variance $(\eta_1+1)(\eta_2+1)-1$.
\end{lemma}
\eat{
\paragraph{Gomory-Hu Tree.}
The importance sampling subroutine of our algorithm makes use of Gomory-Hu trees. The {\em Gomory-Hu tree} \cite{gomory1961multi} of an undirected graph $G = (V, E)$ is a (weighted) tree $Y$ on the same set of vertices $V$ such that for every pair of vertices $s, t\in V$, the minimum $(s, t)$-cut in $Y$ (which is simply the minimum weight edge on the unique $s-t$ path in $Y$) is also a minimum $(s, t)$-cut in $G$, and has the same value.}

We also define the {\em relative bias} of an estimator:
\begin{defn}
The {\em relative bias} of an estimator $X$ for a value $x$ is defined as $\frac{|\E[X]-x|}{x}$. 
\end{defn}

\section{Importance Sampling on a Spanning Tree Packing}
\label{sec:importance}
As stated earlier, our main new algorithmic contribution is an estimator for $\ugp$ in the very reliable case, i.e., $p^\lambda \le O(n^{-3})$. Our goal in this section is to prove \Cref{lem:interface-importance}, which we restate below:

\importance*

We design an importance sampling algorithm that gives an estimator $X$ of $\ugp$ with a relative bias of $O(\log n/\sqrt{n})$ in $m^{1+o(1)}+\tO(n^{1.5})$ time. Since $n > \tilde{\Omega}(\eps^{-2})$ (otherwise we are in the first base case), the relative bias is at most $\eps/3$. It follows by \Cref{lem:mc-sample} that $\tO(\eps^{-2})$ runs of the algorithm gives a $(1\pm\eps/3)$-approximation of $\E[X]$, which in turn is a $(1\pm\eps/3)^2 \in (1\pm\eps)$-approximation of $\ugp$. 


%

\subsection{Dependence of $\zgp$ on Near-Minimum Cuts}

By \Cref{lem:z-approx-u-reliable}, since $p^\lambda \le O(n^{-3})$, we have 
\[
    \left(1 - O\left(\frac{\log n}{\sqrt{n}}\right)\right)\zgp \le \ugp \le \zgp.
\] 
This allows us to focus on approximating $\zgp$ instead of $\ugp$.

Intuitively, when $p^\lambda$ is small, the value of $\zgp$ only depends on the near-minimum cuts because larger cuts scarcely fail. We make this intuition formal below. 
Let $M_k$ be the number of cuts of value $k$. Then, $\zgp$ is defined as
\[ \zgp = \sum_{k\ge \lambda} M_k \cdot p^{k}.\]
By cut counting (\Cref{lem:cut-counting}), we know that $M_k\le n^{2k/\lambda}$. The assumption $p^\lambda \le O(n^{-3})$ implies $p^k \le O(n^{-3k/\lambda})$ for all $k \ge \lambda$. So, $p^k$ decreases much faster than the increase in $M_k$ as $k$ increases; as a consequence, the product $M_k\cdot p^k$ decays rapidly. The next lemma shows that we can truncate the sum $\sum_{k\ge \lambda} M_k\cdot p^k$ at $k = 3.5\lambda$ without significantly distorting its value. To state the lemma, let us define the partial sums:
\[z^{\ge\alpha}=\sum_{C_i:c_i\ge \alpha\lambda}p^{c_i} \qquad \text{and} \qquad  z^{\le\alpha}=\sum_{C_i:c_i\le \alpha\lambda}p^{c_i}.\]

\begin{lemma}\label{lem:p-small-truncate}
If $p^\lambda \le O(n^{-3})$, then $z^{\ge 3.5}\le O(n^{-0.5})\cdot \zgp$, where $z^{\ge 3.5}$ is the expected number of failed cuts of value at least $3.5\lambda$.
\end{lemma}
\begin{proof}
By \Cref{lem:cut-counting}, there are at most $n^{\lfloor 2\alpha\rfloor}$ $\alpha\lambda$-weak cuts. So for each integer $k\ge 7$, the expected number of failed cuts with cut value in $\left[\frac{k}{2}\cdot \lambda, \frac{k+1}{2}\cdot \lambda\right)$ is at most $n^{k}\cdot p^{k\lambda/2}$. Therefore,
\begin{align*}
    \frac{z^{\ge 3.5}}{p^\lambda} &\le \frac{1}{p^\lambda}\sum_{k\ge 7} n^k \cdot p^{k\lambda/2}
    \le n^7 \cdot p^{2.5\lambda}\cdot \left(\sum_{k\ge 0} n^k \cdot p^{k\lambda/2}\right)
    = O\left(n^7 \cdot p^{2.5\lambda}\right)
    = O(n^{-0.5})
\end{align*}
since $p^\lambda = O(n^{-3})$.
Therefore $z^{\ge 3.5}= O(n^{-0.5})\cdot p^\lambda \le O(n^{-0.5})\cdot \zgp$.
\end{proof}

\subsection{Algorithm for Estimating $\zgp$ for $3.5\lambda$-weak Cuts}\label{sec:sample-algorithm}
\Cref{lem:p-small-truncate} implies that for the purpose of getting a $(1\pm\eps)$-approximation of $\zgp$, it suffices to only consider cuts of value $\alpha\lambda$ for $1 \le \alpha \le 3.5$. We first give the high level idea for estimating the contribution of these cuts to $\zgp$. Suppose we choose  a spanning tree of the graph uniformly at random from a collection of $\frac{\lambda}{2}$ edge-disjoint spanning trees. By averaging, we expect to see at most $7$ edges from a $3.5\lambda$-weak cut in this randomly chosen tree. We don't know how to sample a $3.5\lambda$-weak cut directly, but we can sample a cut whose projection on the spanning tree contains at most $7$ edges. Therefore, we can write an estimator that reweights these cuts appropriately to obtain an unbiased sample of $\zgp$ restricted to the $3.5\lambda$-weak cuts. This is the {\em importance sampling} problem that we solve below.

There are two main challenges in this problem. First, even if we had access to a uniform distribution over these $\tO(n^7)$ cuts, we do not have enough time to draw sufficiently many samples to pick every cut. Hence, we need a careful sampling and reweighting algorithm that allows us to sample much fewer cuts but still keep the variance under control. Second, we do not have access to a uniform distribution over these cuts. Instead, we can only sample a tree, obtain a random set of $7$ or fewer edges in the tree, and define the corresponding cut as our sampled cut. This creates a biased distribution over the cuts themselves, and our reweighting must eliminate this bias, again without increasing variance.

Formally, our algorithm for estimating $\zgp$ has the following three steps.

\paragraph{Step 1: Sparsification.}
First, we apply the sparsification algorithm (\Cref{lem:sparsify}) with parameter $\delta=\frac{1}{\log n}$ to get a sparsifier $H$.
%

\paragraph{Step 2: Tree Packing.}
Next, we construct a packing of $\lambda_H$ spanning trees in the sparsifier graph $H$ where each edge appears at most twice:
\begin{lemma}\label{lem:tree-packing}
Given an undirected graph $G$ with min-cut value $\lambda$, we can construct in $\tO(\lambda m)$ time a collection $\mathcal{T}$ of $\lambda$ spanning trees such that every edge appears in at most two trees. 
\end{lemma}
This lemma is well-known and follows, e.g., from Gabow~\cite{Gabow95}; we give a short proof based on Gabow's result in the appendix for completeness.
Note that since $\lambda_H = O(\log^3 n)$, this algorithm runs in $\tO(m)$ time. 

We say that a cut {\em $k$-respects} a tree if there are at most $k$ edges from the cut in the tree. The key property of the spanning tree packing in \Cref{lem:tree-packing} is that every $3.5\lambda$-weak cut will 7-respect at least one tree in the packing:
\begin{lemma}\label{lem:cut-intersect-tree}
Fix any $k\in \{2, 3, \ldots, 7\}$. There exists a large enough constant $\beta$ such that it holds whp that for every cut $C_i$ (in $G$) with $c_i \le \left(k+1-\frac{\beta}{\log n}\right)\cdot \frac{\lambda}{2}$, there exists a tree $T\in\mathcal{T}$ such that $|C_i\cap T|\le k$.
\end{lemma}
\begin{proof}
Let ${\cal C}_k$ be the set of cuts $C_i$ in $G$ satisfying $c_i \le \left(k+1-\frac{\beta}{\log n}\right)\cdot \frac{\lambda}{2}$. Recall that $\delta=\frac{1}{\log n}$ is the sparsification parameter we used when applying \Cref{lem:sparsify}. Thus, for every cut $C_i$ in ${\cal C}_k$, we have $c_i \le \left(k+1-\beta\delta\right)\cdot \frac{\lambda}{2}$. After sparsification, the value of every cut in ${\cal C}_k$ in the sparsifier $H$ is at most $(1+\delta)(k+1 - \beta\delta)\cdot\alpha\cdot\frac{\lambda}{2}$ whp. Moreover, $\lambda_H \ge (1-\delta)\alpha\lambda$ whp. Therefore, the value of every cut in ${\cal C}_k$ in the sparsifier $H$ is at most:
\begin{equation*}
    (1+\delta)(k+1 - \beta\delta)\cdot\alpha\cdot\frac{\lambda}{2} = \frac{(1+\delta)(k+1 - \beta\delta)}{1-\delta}\cdot\frac{\lambda_H}{2} < (k+1)\frac{\lambda_H}{2}, \quad \text{ for } \beta \ge 2(k+1).
\end{equation*}
Since $k \le 7$, we can set $\beta = 16$  to ensure the inequality above. This implies that after sparsification, each cut in ${\cal C}_k$ has value strictly smaller than $(k+1)\frac{\lambda_H}{2}$ in $H$. Since there are $\lambda_H$ trees produced by \Cref{lem:tree-packing} and every edge can appear at most twice, it follows that every cut in ${\cal C}_k$ has strictly less than $k+1$ edges on average across the trees. Therefore, for every cut in ${\cal C}_k$, there is at least one tree containing at most $k$ edges from the cut.
\end{proof}

As a consequence of \Cref{lem:cut-intersect-tree}, it suffices to calculate the contribution to $\zgp$ of all cuts that 7-respect some tree in the tree packing. We call this latter set ${\cal C}_7$. Note that the set $\Cresp$ includes all $3.5\lambda$-weak cuts in $G$, but might include other cuts as well. 
The remainder of the section will design an unbiased estimator of $\zresp = \sum_{C_i\in \Cresp} p^{c_i}$. \Cref{lem:zresp-approx-ugp} shows that the estimator for $\zresp$ is also a (biased) estimator of $\ugp$ with an overall relative bias of $O\left(\frac{\log n}{\sqrt{n}}\right)$.

\begin{lemma}\label{lem:zresp-approx-ugp}
    Assume $p^{\lambda}<O(n^{-3})$. Then $|\zresp-\ugp|\le O\left(\frac{\log n}{\sqrt{n}}\right)\ugp$.
\end{lemma}
\begin{proof}
    \Cref{lem:cut-intersect-tree} gives that $z^{\le 3.5} \le \zresp \le \zgp$. By \Cref{lem:p-small-truncate}, 
    $$|\zresp-\zgp|\le |z^{\le 3.5}-\zgp|\le |z^{\ge 3.5}|\le O(n^{-0.5})\cdot \zgp.$$
    By \Cref{lem:z-approx-u-reliable}, since $p^{\lambda}<O(n^{-3})$, we have $|\zgp-\ugp|\le O\left(\frac{\log n}{\sqrt{n}}\right)\zgp$ and $\zgp =O(\ugp)$. Therefore,
    $$ |\zresp-\ugp|\le |\zresp-\zgp|+ |\zgp-\ugp| \le O\left(\frac{\log n}{\sqrt{n}}\right)\ugp.$$
\end{proof}

\paragraph{Step 3: Unbiased Estimator for $\zresp$ via Importance Sampling.}
Recall that $\zresp$ is a sum of $p^{c_i}$ over the $\tO(n^7)$ cuts $C_i$ in $\Cresp$. This is much fewer than the $O(2^{n})$ cuts overall, but we still cannot afford to directly enumerate all of them. Instead, using importance sampling, we obtain an estimator of $\zresp$ with small variance.

Our estimator $X$ is defined as follows: $X = \frac{p^{c_i}}{q(C_i)}$ with probability $q(C_i)$, for some distribution $q: \Cresp \rightarrow [0, 1]$ that we will define later. This $q$ will be tailored so that we can efficiently sample cuts from $q$, and moreover that given a cut $C_i$, we can efficiently compute $q(C_i)$. In \Cref{lem:importance-ex-var}, we show that $X$ is an unbiased estimator of $\zresp$, and also obtain a bound on the relative variance of $X$ as a function of the distribution $q$:
\begin{lemma}\label{lem:importance-ex-var}
    $X$ is an unbiased estimator of $\zresp$ and its relative variance $\relv[X]$ satisfies whp
    \[\relv[X]\le \frac{1}{\zresp}\cdot \left(\max_{C_i\in \Cresp} \frac{p^{c_i}}{q(C_i)}\right).\]
\end{lemma}
\begin{proof}
The expectation of $X$ is given by
$$\E[X]=\sum_{C_i\in \Cresp} q(C_i)\cdot \frac{p^{c_i}}{q(C_i)} =\sum_{C_i\in \Cresp} p^{c_i} = \zresp.$$
The relative variance of $X$ satisfies
\begin{align*}
\eta[X] \quad
& < \quad \frac{\E[X^2]}{\E[X]^2} 
\quad =  \quad\frac{1}{(\zresp)^2}\sum_{C_i \in \Cresp} q(C_i)\left(\frac{p^{c_i}}{q(C_i)}\right)^2 \\
\quad & \quad \le \quad \frac{1}{\zresp} \left(\max_{C_i\in \Cresp} \frac{p^{c_i}}{q(C_i)}\right)\left(\frac{\sum_{C_i \in \Cresp} p^{c_i}}{\zresp}\right)
\quad = \quad \frac{1}{\zresp} \left(\max_{C_i\in \Cresp} \frac{p^{c_i}}{q(C_i)}\right).\qedhere
\end{align*}
\end{proof}
To minimize the relative variance, we would ideally want $q(C_i) \propto p^{c_i}$ for all $C_i\in \Cresp$. But, this is impossible to ensure exactly because the set $\Cresp$ is unknown and too large to enumerate. Instead, we use a surrogate distribution $q$ to approximate this ideal distribution. The distribution $q$ is defined as the mixture of a set of distributions $q_j$ for $j \in \{1, 2, \ldots, 7\}$, where $j$ can be loosely interpreted as the number of edges in the intersection of the sampled cut $C\in \Cresp$ and some tree $T$ chosen randomly from the packing $\cal T$ given by \Cref{lem:tree-packing}. For $j \ge 2$, the distribution $q_j$ is given by the following process: First, we pick a tree $T$ uniformly at random from the packing $\cal T$ given by \Cref{lem:tree-packing}. Next, pick $j$ edges uniformly at random from $T$ {\em with replacement}. The cut $C$ is then defined as the unique cut in $G$ that intersects $T$ at precisely the chosen edges. (Note that the number of chosen edges might actually be less than $j$ because the sampling is with replacement.) 

We now precisely calculate the values of $q_j(C_i)$ for any cut $C_i$ and any $j\in \{2, 3, \ldots, 7\}$. The following fact is useful for this purpose (we include a proof in the appendix):
\begin{fact}\label{fact:importance-distribution-correct}
    Given a universe $U$ of $N$ elements and a set $A\subseteq U$ of size $\alpha$, if we pick $j$ elements from $U$ uniformly at random with replacement, then the probability that the set of elements picked is precisely $A$ is given by $\frac{\alpha!\, S(j, \alpha)}{N^j}$, where $S(j, \alpha)$ is the Stirling number of the second kind. 
    
    In particular, when $j=2$, the probability is $\frac{\alpha}{N^j}$ for $\alpha\in\{1,2\}$ and 0 otherwise.
\end{fact}

Using \Cref{fact:importance-distribution-correct}, we can now infer that for any cut $C_i\in \Cresp$ and any $j\in \{2, 3, \ldots, 7\}$, we have
\begin{equation}\label{eq:qj}
    q_j(C_i)=\frac{1}{|{\cal T}|}\cdot \sum_{T\in\mathcal{T}} \frac{(|C_i\cap T|)!\cdot S(j, |C_i\cap T|)}{(n-1)^j},
\end{equation}
where $S(j, \alpha)$ is the Stirling number of the second kind.

We will show later that in the sampling process for $q_j$ described above, cuts of value at least $1.5\lambda$ contribute at most $\tO(n^{1.5})$ to the relative variance of the estimator. But, the contribution of $1.5\lambda$-weak cuts to the relative variance of the estimator can be as large as $\Omega(n^2)$. So, we cannot simply repeat the sampling to obtain an estimator with constant relative variance. 

Instead, we mix an additional distribution in $q$ that we call $q_1$. This distribution $q_1$ amplifies the weight on the small cuts, so that they contribute less to the relative variance. Alternatively said, $q_1$ dampens the effect on the relative variance of the very large cuts. This distribution is based on a tree packing on a contracted graph. 

To define the contracted graph, 
we construct a Gomory-Hu tree $Y$ (\Cref{def:gomory-hu}) of our input graph $G$. Let $\tau$ be the $\sqrt{n}$-th smallest edge weight in $Y$. We contract all edges of weight at least $\tau$ in $Y$; the sets of vertices that are contracted into single nodes are now contracted in $G$ as well. This results in a graph $G'$ and its Gomory-Hu tree $Y'$, both on $n' \le \sqrt{n}$ vertices. Next, we repeat the sparsification (\Cref{lem:sparsify}) and tree packing (\Cref{lem:tree-packing}) steps on $G'$ to obtain a tree packing $\mathcal{T}'$ on a sparsifier $H'$ of $G'$. By the same argument as for $\mathcal{T}$, we have $|\mathcal{T}'| = O(\log^3 n) = \tO(1)$. Now we define $q_1$ as follows: First, pick a tree $T'$ in $\mathcal{T}'$ uniformly at random. Then, choose 2 edges {\em with replacement} uniformly at random from $T'$. Finally, define the sampled cut $C$ as the unique cut in $G'$ (and therefore in $G$) that intersects $T'$ at precisely the chosen edges. (Note that because of sampling with replacement, the number of edges in the intersection can be either 1 or 2.) 
Using \Cref{fact:importance-distribution-correct}, we can precisely state the probability $q_1(C_i)$ for any cut $C_i\in \Cresp$:
\begin{equation}\label{eq:q1}
    q_1(C_i)=\frac{1}{|{\cal T}'|} \cdot \sum_{T'\in\mathcal{T}':|C_i\cap T'|\in\{1, 2\}}\frac{|C_i\cap T'|}{(n'-1)^2}.
\end{equation}

We have given a set of distributions $q_1, q_2, \ldots, q_7$ for defining the cut $C$. Finally, we combine these distributions uniformly: namely, we choose an index $j\in \{1, 2, \ldots, 7\}$ uniformly at random and then choose $C$ according to distribution $q_j$. Thus, $q(C_i) = \frac{1}{7}\sum_{j=1}^7 q_j(C_i)$ for $C_i\in \Cresp$.

Recall that in \Cref{lem:importance-ex-var}, we bounded the relative variance of the estimator by $\frac{1}{\zresp}\cdot \max_{C_i\in \Cresp} \frac{p^{c_i}}{q(C_i)}$. 
We start by obtaining a bound on $\frac{1}{q(C_i)}$ by individually bounding $\frac{1}{q_j(C_i)}$ for every $j\in \{1, 2, \ldots, 7\}$.

\begin{lemma}\label{lem:sample-wt}
If a cut $C_i$ 2-respects some tree in $\mathcal{T}'$, then $\frac{1}{q(C_i)} \le \tO(n)$. For $j\in \{2, 3, \ldots, 7\}$, if a cut $C_i$ $j$-respects some tree in $\mathcal{T}$, then $\frac{1}{q(C_i)}\le \tO(n^j)$.
\end{lemma}
\begin{proof}
    First, consider a cut $C_i$ that 2-respects a tree $T'\in \mathcal{T}'$. Since $|C_i\cap T'|\in\{1, 2\}$, we have 
    $$q(C_i) \ge \frac 17 \cdot q_1(C_i) \ge \frac{1}{7\cdot |\mathcal{T}'|\cdot (n'-1)^2}$$
    We have $|\mathcal{T}'| = O(\log^3 n)$ and $n' \le \sqrt{n}$. Hence, $\frac{1}{q(C_i)} \le \tO(n)$.
    
    For $C_i$ that $j$-respects $T\in \mathcal{T}$ for some $j\in \{2, 3, \ldots, 7\}$, we have $1\le |C_i\cap T|\le j$. Note that $S(j, \alpha)\ge 1$ for $1\le\alpha\le j$. Therefore,
    $$q(C_i) \ge \frac 17 \cdot q_j(C_i) \ge \frac{1}{7\cdot |\mathcal{T}|\cdot (n-1)^j}$$
    Since $|\mathcal{T}|=O(\log^3 n)$, we have $\frac{1}{q(C_i)} \le  \tO(n^j)$.
\end{proof}

Next, we bound the expression $\frac{1}{\zresp}\cdot \frac{p^{c_i}}{q(C_i)}$ for every cut $C_i\in \Cresp$, thereby bounding the relative variance using \Cref{lem:importance-ex-var}.

\begin{lemma}\label{lem:importance-relvar}
Assume $p^{\lambda}\le O(n^{-3})$. For every cut $C_i\in \Cresp$, we have  $\frac{1}{\zresp}\cdot \frac{p^{c_i}}{q(C_i)}\le \tO(n^{1.5})$.
\end{lemma}
\begin{proof}
By \Cref{lem:sample-wt}, $\frac{1}{q(C_i)} = \tO(n^7)$ for every cut $C_i\in \Cresp$. This is quite a loose bound but it already suffices for large cuts, namely when $c_i \ge 3.5\lambda$. In this case, $p^{c_i-\lambda}\le p^{2.5\lambda} \le O(n^{-7.5})$. Therefore, 
\[
\frac{1}{\zresp}\cdot \frac{p^{c_i}}{q(C_i)} \le \frac{p^{c_i-\lambda}}{q(C_i)}\le O(n^{-7.5})\cdot \tO(n^7) \le O(1).
\]

Next, suppose the value of $C_i$ satisfies
$(j-\beta\delta)\cdot \frac{\lambda}{2} < c_i \le (j+1-\beta\delta)\cdot \frac{\lambda}{2}$ for some $j \in \{3, 4, \ldots, 7\}$, where  $\delta=\frac{1}{\log n}$ and $\beta$ is the constant in \Cref{lem:cut-intersect-tree}. Then, from \Cref{lem:cut-intersect-tree}, we know that $C_i$ $j$-respects some tree $T$ in the packing $\cal T$. Thus, $\frac{1}{q(C_i)} = \tO(n^j)$ by \Cref{lem:sample-wt}. 
Since $p^{\lambda} \le O(n^{-3})$ and $\frac{j-\beta\delta}{2}>1$, we get 
\[
\frac{p^{c_i}}{\zresp} 
\le p^{c_i-\lambda}
\le p^{ (j-\beta\delta)\cdot \frac{\lambda}{2} - \lambda}
= \left(p^{\lambda}\right)^{\frac{j-\beta\delta}{2}-1}
\le \left(O(n^{-3})\right)^{\frac{j-\beta\delta}{2}-1}
\le O(n^{3-1.5j}).
\]
So $\frac{p^{c_i}}{\zresp}\cdot \frac{1}{q(C_i)} \le \tO(n^{3-0.5j})\le \tO(n^{1.5})$ since $j\ge 3$.

Note that the previous case covers all $C_i\in \Cresp$ satisfying $c_i > (3-\beta\delta) \cdot \frac{\lambda}{2}$.
The remaining case is $c_i \le (3-\beta\delta) \cdot \frac{\lambda}{2}$. By \Cref{lem:cut-intersect-tree}, $C_i$ 2-respects some tree $T$ in the packing  $\mathcal{T}$, and hence $\frac{1}{q(C_i)} = \tO(n^2)$ by \Cref{lem:sample-wt}. Now, there are two subcases depending on whether $c_i \ge \tau$ or $c_i < \tau$. (Recall that $\tau$ is the value of $\sqrt{n}$th smallest weight edge in the Gomory-Hu tree $Y$ of $G$.) 

Consider $c_i \ge \tau$. In this case, $Y$ (and therefore $G$) has at least $\sqrt{n}$ cuts of value at most $\tau$. Since $\tau \le c_i < 1.5\lambda$, it follows that all these $\sqrt{n}$ cuts are in $\Cresp$. For each such cut $C$, we have $p^{|C|} \ge p^{\tau}$. Therefore, $\zresp\ge \sqrt{n}\cdot p^\tau$. Hence,
$$\frac{1}{q(C_i)}\cdot \frac{p^{c_i}}{\zresp}\le \tO(n^2) \cdot \frac{p^{\tau}}{\zresp} \le \tO(n^{1.5}).$$

Next, consider $c_i < \tau$. In this case, we claim that the cut $C_i$ is preserved in $G'$. Since the edges that are contracted in $Y$ to form $Y'$ are all of value at least $\tau$, it follows that any pair of vertices in $G$ that are contracted to the same node in $G'$ must be $\tau$-connected. On the other hand, every edge in $C_i$ is in a cut of value strictly less than $\tau$, and therefore, the vertices at the ends of the edge are not $\tau$-connected. In particular, note that the minimum cuts in $G$ are also preserved in $G'$, and therefore, the min-cut value of $G'$ is also $\lambda$. Therefore, by \Cref{lem:cut-intersect-tree}, we can conclude that $C_i$ 2-respects some tree in $\mathcal{T}'$. It follows by \Cref{lem:sample-wt} that $\frac{1}{q(C_i)} = \tO(n)$. We can now bound $\frac{1}{q(C_i)}\cdot \frac{p^{c_i}}{\zresp} \le \frac{1}{q(C_i)} = \tO(n)$.
\end{proof}

Finally, we combine the results to give a proof of \Cref{lem:interface-importance}.

\begin{proof}[Proof of \Cref{lem:interface-importance}]
    We show that the average of multiple samples of our estimator $X$ satisfies the lemma.

    By \Cref{lem:importance-ex-var}, $\E[X]=\zresp$.
    Using \Cref{lem:importance-relvar} and \Cref{lem:importance-ex-var}, we conclude that the relative variance of our estimator is $\tO(n^{1.5})$. Therefore, by repeating the sampling algorithm $\tO(n^{1.5})$ times, and using \Cref{lem:relvar-multiply}, we obtain an unbiased estimator for $\zresp$ with relative variance $O(1)$. By \Cref{lem:zresp-approx-ugp}, this estimator has relative bias $O\left(\frac{\log n}{\sqrt{n}}\right)$ for $\ugp$.
    
    We will show in \Cref{thm:running-time} that the algorithm has preprocessing time $m^{1+o(1)}$ and each sample takes $\tO(1)$ time. We take $\tO(n^{1.5})$ samples, so the running time is $m^{1+o(1)}+\tO(n^{1.5})$.
\end{proof}


\subsection{Running Time of the Algorithm}
We now show that after $m^{1+o(1)}$ preprocessing time, we can sample the estimator described in the previous section in $\tO(1)$ time. Since the estimation algorithm is repeated $\tO(n^{1.5})$ times, this yields an overall running time of $m^{1+o(1)} + \tO(n^{1.5})$.

\paragraph{Preprocessing of the algorithm.}
Preprocessing involves the following steps: computing a Gomory-Hu tree (takes $m^{1+o(1)}$ time for an unweighted graph by \cite{Abboud22}), sparsification using uniform probabilities (takes $O(m)$ time by \Cref{lem:sparsify}), tree packing (takes $\tO(m\lambda_H) = \tO(m)$ time by \Cref{lem:tree-packing} since $\lambda_H = \tO(1)$), and contractions on the Gomory-Hu tree (takes $O(m)$ time using, e.g., breadth-first search). Therefore, preprocessing takes a total of $m^{1+o(1)}$ time.

\paragraph{Data structure for cut queries.}
We are left to show that we can sample the estimator in $\tO(1)$ time. For this purpose, 
we describe a data structure with respect to any graph $G$ and a tree $T$ defined on the vertices $V$ of $G$. (In general, we do not require $T$ to be a subgraph of $G$.) First, we define a mapping of $G = (V, E)$ into a set of points in $\RR^2$. To define the mapping, we define an order on the vertices $V$ in $G$. We create an Euler tour of $T$ starting at an arbitrary vertex; call the resulting vertex sequence $\euler(T)$. Next, we order vertices by their first appearance in $\euler(T)$; we call this order the {\em preorder} sequence of the vertices and denote it $v_1, v_2, \ldots, v_n$. Next, for every edge $e = (v_i, v_j)$ in $G$ where $i < j$, we map it to the point $(i, j)$ on $\RR^2$. 

Our data structure, which we denote $\bS_{G, T}$, supports 2-dimensional {\em orthogonal range queries} on the set of points corresponding to the edges in $G$. Namely, given a pair of intervals $U$ and $W$ in $\RR$, the data structure reports the number of points in the rectangle $U \times W$; we denote this count $\bS_{G, T}(U, W)$. In particular, suppose $U$ and $W$ correspond to {\em disjoint} sets of vertices that are contiguous in the preorder sequence and $U$ comes before $W$ in the sense that for every $v_i\in U, v_j\in W$, we have $i < j$. Then, $\bS_{G, T}(U, W) = |\{(v_i, v_j)\in E: v_i\in U, v_j\in W\}|$. In other words, it counts the number of edges between the vertex sets $U$ and $W$. This data structure can be implemented using standard results in computational geometry, e.g., \cite{Lueker78}. For a set of $m$ points, the data structure has a preprocessing time of $\tO(m)$ and supports orthogonal range queries (i.e., reports $\bS_{G, T}(U, W)$ for given $U, W$) in $O(\log^2 m)$ time.

On top of this data structure, we need to maintain some additional information. Recall the Euler tour $\euler(T)$ of tree $T$. Note that every edge is traversed twice in $\euler(T)$. We label the traversal of an edge by the number of distinct vertices in $V$ that we have seen so far in the Euler tour. This means that for any $i$, the Euler tour visits all edges that have a label of $i$ between the first occurrences of $v_i$ and $v_{i+1}$ in $\euler(T)$. Conversely, each edge gets two labels $i$ and $j$ and is visited between the first occurrences of $v_i, v_{i+1}$ and $v_j, v_{j+1}$ in $\euler(T)$. We denote the two labels for edge $e$ by the (two-element) set $L(e)$ and collectively, $\bL_T = \{L(e): e\in T\}$. Note that we can compute $\bL_T$ in $O(n)$ preprocessing time by a single traversal of $\euler(T)$.

Let $\chi$ be a set of edges in $T$. There is a unique cut in $T$ such that the cut edges are exactly those in $\chi$. We denote the vertex bipartition of this cut by $C(\chi)$. Since $G$ and $T$ are defined on the same set of vertices, the vertex bipartition $C(\chi)$ also induces a cut in $G$; overloading notation, we call this cut $C(\chi)$ as well. Note that in particular, if $T$ is a subgraph of $G$, then $\chi = T\cap C(\chi)$.

Using our data structures $\bS_{G, T}$ and $\bL_T$, we now show that the following cut query can be answered efficiently: {\em given an edge set $\chi$ in $T$ of constant size, calculate the cut value of the vertex bipartition $C(\chi)$ in $G$.}
\begin{lemma}
\label{lem:range}
    Suppose we have the data structures $\bS_{G, T}$ and $\bL_T$ for a graph $G$ and a spanning tree $T$ of $G$. Now, for any set $\chi$ of a constant number of edges in $T$, the value of the cut $C(\chi)$ in $G$ can be computed using a constant number of orthogonal range queries on $\bS_{G, T}$. Therefore, the total query time for $\chi$ is $\tO(1)$.
\end{lemma}
\begin{proof}
Let $\bL_T(\chi)$ be the multiset of labels $L(e)$ corresponding to the edges $e\in \chi$. Since each edge has two labels and $|\chi| = O(1)$, this multiset has constant size. Next, we sort the labels in $\bL_T(\chi)$; denote this sorted sequence $\ell_1, \ell_2, \ldots, \ell_{2|\chi|}$. We append this list with $\ell_0 = 0$ and $\ell_{2|\chi|+1} = n$. For any pair of consecutive labels $\ell_i, \ell_{i+1}$ for $i = 0, 1, \ldots, 2|\chi|$ in this appended list, we define a vertex set $V_i = \{v_j: \ell_i < j \le \ell_{i+1}\}$. Note that it is possible that $\ell_i = \ell_{i+1}$, in which case $V_i = \emptyset$. Now, define $V_{\odd} = V_1\cup V_3\cup \ldots \cup V_{2|\chi|-1}$ and $V_{\even} = V_0\cup V_2\cup \ldots \cup V_{2|\chi|}$. Note that although the sets $V_{\odd}$ and $V_{\even}$ can be superconstant in size, they can be represented and computed from the labels in $\bL_T(\chi)$ in $O(1)$ time by denoting each interval $V_i$ by its two endpoints $\ell_i+1$ and $\ell_{i+1}$.

Finally, we note that the vertex bipartition $C(\chi)$ is precisely given by the sets $V_{\odd}$ and $V_{\even}$. To see this, consider the Euler tour $\euler(T)$ and mark every edge in $\chi$ on it. Now, the vertices switch from one side of $C(\chi)$ to the other side every time we traverse a marked edge. This corresponds to taking the odd and even indexed sets $V_i$, which is exactly how we defined the sets $V_{\odd}$ and $V_{\even}$. Finally, we take each pair of sets $V_i, V_j$ where $V_i\in V_{\odd}$ and $V_j\in V_{\even}$ and run the query $\bS_{G, T}(V_i, V_j)$ on the data structure $\bS_{G, T}$. Note that this is only $O(1)$ number of queries, and each query takes $O(\log^2 m)$ time by the properties of $\bS_{G, T}$. Finally, we return the sum of the values returned by these queries as the number of edges in cut $C(\chi)$ in graph $G$.
\end{proof}

Specifically, we instantiate the following data structures $\bS$ and $\bL$ for our purpose. 
First, we construct $\bS_{G, T}$ for every tree $T\in {\cal T}\cup {\cal T}'$ (with $G$ as the input graph). 
Note that the data structure requires $T$ to be defined on the same set of vertices as $G$, which does not hold for the trees $T'\in {\cal T}'$. To resolve this, we expand the contracted vertices in $G'$ into connected components in $G$ and connect each component by an arbitrary spanning tree. By doing this, the tree $T'$ is expanded into a tree defined on the vertices $V$. Moreover, this step will not change the intersection of $T'$ with any cut in $G'$. Therefore, we can use the expanded tree to build $\bS_{G, T'}$ and report the correct cut value of $C(\chi)$ for any $\chi$ defined on $T'$.
We also build a data structure $\bS_{T_1, T_2}$ for every pair of trees $T_1, T_2 \in {\cal T}\cup {\cal T}'$. Note that after expansion of the trees in ${\cal T}'$, all the trees $T_1, T_2$ are on the same set of vertices $V$; thus the data structures are well-defined. In addition to these data structures, we also construct the label sets $\bL_T$ for all trees $T\in {\cal T} \cup {\cal T}'$.
Note that the total number of data structures is $\tO(1)$, which implies that the preprocessing time increases only by a $\tO(1)$ factor.

\paragraph{Running time of the sampling algorithm.}
First, we compute the running time for producing a single sample for the estimator $X$ defined in Step 3 of \Cref{sec:sample-algorithm} using the data structures $\bS_{G, T}$ and $\bL_T$ that we defined above.
The first step of the sampling algorithm is to choose a specific distribution $q_j$ from $j\in \{1, 2, \ldots, 7\}$ uniformly at random in $O(1)$ time.
Next, the algorithm chooses a tree $T$ uniformly at random from the packing $\cal T$ (or ${\cal T}'$ for $q_1$) in $\tO(1)$ time.
Third, the algorithm chooses a set $\chi$ of $j$ edges from $T$ uniformly at random with replacement. This takes $\tO(1)$ time since $j$ is a constant.

The only involved step is to calculate the estimator $X$ at this point. Recall that $X$ is defined as $\frac{p^{|C(\chi)|}}{q(C(\chi))}$. First, we compute the cut value $|C(\chi)|$ by querying $\bS_{G, T}$ in $\tO(1)$ time by \Cref{lem:range}. We are left to compute $q(C(\chi))$, which only requires computing  $|C(\chi)\cap T'|$ for every tree $T'\in {\cal T}\cup {\cal T}'$ (where $C(\chi)$ is defined on $G$) by using \Cref{eq:qj,eq:q1}. (Note that the Stirling numbers in \Cref{eq:qj} can be retrieved from a constant-sized table since $j\le 7$.) To compute $|C(\chi)\cap T'|$, we use the data structure $\bS_{T', T}$ to obtain the cut value of $C(\chi)$ in $T'$, which is $|C(\chi)\cap T'|$ where $C(\chi)$ is defined on graph $G$. Overall, this requires $\tO(1)$ queries of the data structures, each of which takes $\tO(1)$ time by \Cref{lem:range}.

We have established the following theorem:
\begin{theorem}\label{thm:running-time}
    There is an algorithm that takes $m^{1+o(1)}$ preprocessing time and $\tO(1)$ query time to produce a sample of the estimator $X$.
\end{theorem}


\section{Recursive Contraction}
\label{sec:contraction}


In this section, we design a recursive contraction algorithm that solves the moderately reliable case, that is, $4n^{-3} \le p^\lambda \le n^{-0.5}$. 
Our goal is to show \Cref{lem:interface-contract}, which we restate below:

\contract*

\eat{
We will show the following lemma, which implies \Cref{thm:main}.
\begin{lemma}\label{lem:interface-contract}
    The algorithm outputs an estimator of $\ugp$ with bias at most $0.1\eps\cdot \ugp$ and second moment at most $\log^{O(1)}n\cdot (\ugp)^2$. The running time is $m^{1+o(1)}+\tO(n^{1.5}\eps^{-1})$.
\end{lemma}
\begin{proof}
    The lemma follows \Cref{lem:contract-bias,lem:contract-var,lem:contraction-runtime}.
\end{proof}

\begin{proof}[Proof of \Cref{thm:main}]
Let $X$ be the estimator output by the algorithm.
From \Cref{lem:interface-contract}, we have $$\eta[X]<\frac{\E[X^2]}{\E[X]^2}\le \frac{\log^{O(1)}n\cdot (\ugp)^2}{(1-0.1\eps)^2(\ugp)^2}\le \tO(1).$$

By \Cref{lem:mc-sample}, we can run the algorithm $\tO\left(\frac{\eta[X]}{\eps^{2}}\right) =\tO(\eps^{-2})$ times to get a $(1+\frac{\eps}{2})$-approximation of $\E[X]$ whp.
Because $\E[X]$ is a $(1+0.1\eps)$-approximation of $\ugp$, the aggregated estimator is a $(1+\eps)$-approximation of $\ugp$.
Each run takes $m^{1+o(1)}+\tO(n^{1.5}\eps^{-1})$ time, so the overall running time is $m^{1+o(1)}\eps^{-1}+\tO(n^{1.5}\eps^{-3})$.
\end{proof}
}

\eat{
Our goal is to show the following theorem: 
\begin{theorem}\label{thm:running-time-with-lambda}
There is an algorithm that outputs a $(1+\epsilon)$-approximation of $u_G(p)$ in $\tilde O(m^{1+o(1)}+n^{1.5}\lambda^{1+o(1)})$ time.
\end{theorem}
Combining it with the $\lambda=\tilde O(1)$ bound after sparsification (\Cref{sec:sparsify}), we obtain the desired bound of $\tilde O(m^{1+o(1)}+n^{1.5})$ and establish \Cref{thm:main}.
}

\subsection{Description of the Algorithm}\label{sec:contraction-algorithm}

We use a parameter $\gamma$ to measure how many near-minimum cuts contribute significantly to $\zgp$. Intuitively, if there are very few such cuts, then the average degree is much larger than $\lambda$. In this case, random contraction will significantly decrease the size of the graph. On the other hand, if there are many near-minimum cuts, then the variance of the estimator $u_H(p/q)$ will be small. Our definition of $\gamma$ allows us to smoothly interpolate between these two extreme cases.

We choose $\gamma$ as follows. First, compute the Gomory-Hu tree $Y$ of $G$ (\Cref{def:gomory-hu}). Let $S_k$ be the sum of the $k$ smallest weights among the tree edges. Then,
\[
    \gamma = \min\left\{\lambda,\, \frac 34 \cdot \frac{S_k}{k}\right\} \text{ for } k=2^{-2/3}\cdot n.
\]    

Set $q$ so that $q^{\gamma} = \nicefrac 12$. The algorithm samples two graphs independently $H_1,H_2\sim G(q)$ and returns the average of $u_{H_1}(p/q)$ and $u_{H_2}(p/q)$, which are computed recursively. Note that $G(q)$ represents (the distribution of) the random graph generated by contracting each edge in $G$ independently with probability $1-q$. 

It will be important to ensure that $m \le n^{1.5-\xi}$ for some small enough constant $\xi > 0$ whenever we apply the recursive contraction step. If this condition is violated, we perform two algorithmic steps that restore the condition in expectation. The first step is to reduce the minimum cut of the graph to $\tilde{\lambda} = O(\log^3 n)$ if it is larger. For this purpose, we use the (standard) {\em sparsification} algorithm  in \Cref{sec:sparsify} with sparsification parameter $\delta = \frac{1}{\log n}$. To control the additional variance caused by this sparsification step, we also branch whenever we sparsify. Namely, instead of creating a single sparsifier, we create two independent sparsifier graphs $\tG_1$ and $\tG_2$ from graph $G$ using the same sparsification parameter. We now recurse on these graphs. In particular, if any of the graphs satisfies one of the base cases given below, then we run the corresponding base case algorithm. If not, then we must perform recursive contraction. This means that we draw two independent random samples each from $\tG(q)$ for any such graph $\tG$ produced after sparsification. Let $\tilde{m}$ be the number of edges in any graph drawn from $\tG(q)$. The important property that we we will establish (using \Cref{lem:contraction-size-bound}) is that $\E[\tilde{m}] = O(n\tilde{\lambda}) = \tO(n)$. Now, the algorithm recurses on this graph. In other words, if this new graph has $\tilde{n}$ vertices, then we check if $\tilde{m} > \tilde{n}^{1.5-\xi}$ and continue the algorithm.

There are three base cases in the recursion:
\begin{enumerate}
    \item When $n$ is less than some sufficiently large constant (polynomial in $\eps^{-1}$, will be decided in \Cref{lem:contract-bias}), we use Karger's algorithm (\Cref{thm:karger-estimator}) 
    to compute an unbiased estimator of $u_G(p)$ with relative variance $O(1)$. 
    \item When $p\ge \theta$ or $p^\lambda >  n^{-0.5}$, we use Monte Carlo sampling (\Cref{sec:mc}) to get an unbiased estimator of $u_G(p)$ with relative variance $O(1)$ (\Cref{lem:mc,lem:mc-2step}).
    \item When $p^\lambda < 4n^{-3}$, we use importance sampling on a spanning tree packing (\Cref{sec:importance}) to get an estimate of $z_7(p)$, which is a biased estimator for $\ugp$ with bias $O\left(\frac{\log n}{\sqrt{n}}\right)\cdot \ugp$ and relative variance $O(1)$ (\Cref{lem:interface-importance}).
\end{enumerate}


\eat{
\begin{lemma}\label{lem:contract-final-output}
    Assume $|\E[X]-u|\le 0.1\eps^2 u$ and $\E[X^2] \le \log^K n \cdot u^2$ for constants $\eps \in (0, 1)$ and $K>0$. By first taking the average of $O\left(\frac{\log^K n}{\eps^2}\right)$ independent samples, and then taking the median over $O(\log n)$ such averages, we get a $(1\pm \eps)$-approximation of $u$ whp.
\end{lemma}
\begin{proof}
    Let $\tX$ be the average of $k=\frac{20\log^K n}{\eps^2}$ i.i.d.\ samples of $X$. Note that $\E[\tX] = \E[X]$.
    \begin{align*}
        \E[\tX^2] 
        &= \frac{1}{k}\cdot \E[X^2] + \frac{k-1}{k}\cdot \E^2[X]
        \le \frac{\eps^2 u^2}{20} + (1+0.1\eps^2)^2 \cdot u^2
        \le u^2 + 0.26\eps^2 u^2.\\
        \E[(\tX-u)^2] 
        &= \E[\tX^2] + u^2 -2 u \cdot \E[X]
        \le u^2 + 0.26 \eps^2 u^2 + u^2 - 2(1-0.1\eps^2) u^2
        \le 0.5\eps^2 u^2.
    \end{align*}
    By Markov's inequality,
    \[\Pr[|\tX-u|\ge \eps u] = \Pr[(\tX-u)^2 \ge \eps^2 u^2] \le \frac{\E[(\tX-u)^2]}{\eps^2u^2} \le \frac 12.\]
    If the median falls outside $(1\pm \eps)u$, then half of the samples of $\tX$ fall outside. If we take $2d\log n$ samples, then the failure probability is at most $(\frac 12)^{d\log n} = n^{-d}$.
\end{proof}
}

\subsection{Properties of the Parameter $\gamma$}
The results in this section characterize important properties of $\gamma$ that allow us to simultaneously relate $\gamma$ to the shrinkage of the graph under random contractions and the variance of the recursive estimator for unreliability. 
\begin{fact}\label{fact:gamma-lambda-bound}
$\frac 34 \lambda \le \gamma \le \lambda$.
\end{fact}
\begin{proof}
Recall that $S_k$ is the sum of the $k$ edges with minimum weight in the Gomory-Hu tree and $\gamma=\min\left\{\lambda, \frac 34\cdot \frac{S_k}{k}\right\}$ for $k=2^{-2/3}n$. If $\gamma=\lambda$, the statement is trivial. Next, assume $\gamma = \frac34\cdot \frac{S_k}{k}\le \lambda$. Because each Gomory-Hu tree edge represents a cut in $G$ whose value is at least that of the minimum cut $\lambda$, we have $S_k \ge k\lambda$. Hence, $\gamma = \frac34\cdot \frac{S_k}{k} \ge \frac 34 \lambda$.
\end{proof}

\begin{lemma}\label{lem:gamma-z-bound}
    Assume $p^\lambda \ge 4n^{-3}$. Then $\zgp\ge p^\gamma$.
\end{lemma}
\begin{proof}
If $\gamma=\lambda$, then $\zgp \ge p^\lambda$ is trivial. Next assume $\gamma=\frac 34\cdot \frac{S_k}{k}$.
Consider the partial sum of the $k=2^{-2/3}\cdot n$ smallest cuts represented by edges in the Gomory-Hu tree. Let their cut values be $(a_1, \ldots, a_k)$. 
\begin{equation}\label{eq:z}
\zgp \ge \sum_{i=1}^k p^{a_i}
= k \left(\frac 1k  \sum_{i=1}^k p^{a_i}\right)
\ge k\cdot p^{\frac 1k \sum_{i=1}^k a_i} \text{(by convexity)}
= kp^{S_k/k}
= 2^{-2/3}\cdot n\cdot p^{\frac 43 \gamma}.
\end{equation}
Since $p^\lambda \ge 4n^{-3}$, we have $p^{-\lambda/3} \le 2^{-2/3} \cdot n$. 
Therefore, from \Cref{eq:z}, we get
\[\zgp\ge p^{\frac 43 \gamma -\frac 13 \lambda} \ge p^{\frac 43 \gamma -\frac 13 \gamma} \text{(since $\gamma \le \lambda$)} = p^\gamma. \qedhere \]
\end{proof}

The next two claims are crucial in our analysis of the shrinkage of the graph under random contractions.
\Cref{lem:gamma-deg-bound} gives a lower bound on the number of edges in the graph relative to the number of vertices as we randomly contract edges.
This can be compared to the bound of $\frac{|E(H)|}{|V(H)|}\ge \frac{1}{2}\lambda$ used in \cite{Karger20}. The latter bound can be derived from the fact that the value of the minimum cut in $H$, obtained by contracting edges in $G$, is at least $\lambda$, which is the value of the minimum cut in $G$.  Therefore, the degree of every vertex in $H$ is at least $\lambda$. In comparison, our use of the parameter $\gamma$ allows us to derive a bound of $\frac{|E(H)|}{|V(H)|-1}\ge \frac{2}{3}\gamma$ on the shrinkage of the graph. To compare the two bounds, note that $\frac 23 \gamma \ge \frac 12 \lambda$ since $\gamma \ge \frac 34 \lambda$ by \Cref{fact:gamma-lambda-bound}. Thus, our bound is always (asymptotically) better, which will be crucial in obtaining the better running time bound.
\eat{
\begin{fact}\label{fact:gamma-deg-bound}
    For any $k'\ge 2^{-2/3}\cdot n$, we have $\gamma \le \frac 34 \frac{S_{k'}}{k'}$.
\end{fact}
\begin{proof}
    Because $S_i$ is the prefix sum of a nondecreasing sequence, for $k'\ge k$ we have $\frac{S_{k'}}{k'}\ge \frac{S_{k}}{k} \ge \frac 43 \gamma$.
\end{proof}
}

\begin{lemma}\label{lem:gamma-deg-bound}
    Suppose $H$ is formed by contraction from $G$, and $|V(H)| > 2^{-2/3}\cdot n$. Then 
    \[\frac{|E(H)|}{|V(H)|-1} \ge \frac 23 \gamma.\]
\end{lemma}
\begin{proof}
Root the Gomory-Hu tree $Y$ at an arbitrary vertex $r$. Let $\text{wt}_Y(\cdot)$ denote edge weights in $Y$. Let $k'=|V(H)| > k = 2^{-2/3}\cdot n$. Let $\phi:V(G)\to V(H)$ be the contraction map that takes each vertex of $G$ to its contracted node in $H$. Note that the nodes in $H$ correspond to a partition of the vertices in $G$. Overloading notation, we also use $\phi$ to map a vertex in $V(G)$ to the subset of this partition that the vertex belongs to. 

For any $W\in V(H)\setminus\{\phi(r)\}$, pick an arbitrary vertex $w \in V(G)$ such that $\phi(w)=W$. Since $\phi(w)\ne \phi(r)$, the tree path from $r$ to $w$ in $Y$ contains at least one tree edge $(u_W, v_W)$ (where $u_W$ is the parent of $v_W$ in $T$) such that $\phi(u_W)\ne W$ and $\phi(v_W)=W$. Let $T'$ be the collection of such edges, i.e., $T' = \{(u_W, v_W): W\in V(H)\setminus \{\phi(r)\}\}$. Note that for $W \not= W'$, we have $(u_W, v_W) \not= (u_{W'}, v_{W'})$ since $\phi(v_W) = W$ and $\phi(v_{W'}) = W'$. Thus, $|T'|= k'-1$.

For any $W\in V(H)\setminus \{\phi(r)\}$, the edge $(u_W, v_W)$ represents a minimum $(u_W, v_W)$ cut in $G$. Note also that the degree cut of $W$ in $H$ (denote its value $\deg_H(W)$) is a $(u_W, v_W)$ cut in $G$. Thus, $\text{wt}_Y(u_W, v_W)\le \deg_H(W)$. Taking the sum over edges in $T'$ gives
\begin{equation}\label{eq:gh1}
\sum_{e\in T'}\text{wt}_Y(e)
    \le \sum_{W\in V(H)\setminus\{\phi(r)\}}\deg_H(W)
    \le \sum_{W\in V(H)}\deg_H(W).
\end{equation}

Recall that $S_{k'-1}$ denotes the the sum of the smallest $k'-1$ weights of edges in $Y$. Hence, we have $\sum_{e\in T'}\text{wt}_Y(e) \ge S_{k'-1}$. Now, note that $S_{k'-1}$ is the prefix sum of a nondecreasing sequence, where $k'-1 \ge k$. Thus, $\frac{S_{k'-1}}{k'-1}\ge \frac{S_{k}}{k} \ge \frac 43 \gamma$ by the definition of $\gamma$. Putting these together, we get 
\begin{equation}\label{eq:gh2}
\frac 43 \gamma \le \frac{S_{k'-1}}{k'-1} \le \frac{1}{k'-1}\cdot \sum_{e\in T'}\text{wt}_Y(e).
\end{equation}


Combining \Cref{eq:gh1,eq:gh2}, we get
\[
    \frac 23 \gamma
    \stackrel{(\ref{eq:gh2})}{\le} \frac{\nicefrac 12}{k'-1}\cdot\sum_{e\in T'}   \text{wt}_Y(e)
    \stackrel{(\ref{eq:gh1})}{\le} \frac{\nicefrac 12}{k'-1}\cdot \sum_{W\in V(H)}\deg_H(W)
    = \frac{|E(H)|}{k'-1} \quad\text{(since $\sum_{w\in V(H)}\deg_H(w) = 2|E(H)|$)}.  \qedhere
\]
\end{proof}

An important consequence of the previous lemma (\Cref{lem:gamma-deg-bound}) is that the number of vertices in the graph after a contraction step can be bounded as follows:

\begin{lemma} \label{lem:contract-size-dec}
$H\sim G(q)$ has vertex size at most $(2^{-2/3} +n^{-0.1})\cdot n$ whp.
\end{lemma}
\begin{proof}
The random graph $G(q)$ can be modeled by the following continuous-time random process. Initially, the graph is empty. Edges $e\in E$ arrive at times $t_e$, which are i.i.d.\ exponential variables with rate 1. 

Let us define a graph $G_t$ by contracting all edges in $G$ that arrive before time $t$. Notice that $G_{-\ln q}$ is exactly $G(q)$. Let $n_t$ and $m_t$ be the number of vertices and edges in $G_t$ respectively. 
We consider how the vertex size $n_t$ evolves with $t$. Let $t(k)$ be the first time that $n_t\le k$. Initially $n_0=n$ and $t(n)=0$. Now, $n_t$ decreases by 1 whenever an uncontracted edge arrives. Since there are $m_t$ uncontracted edges at time $t$, the earliest arriving time follows an exponential distribution of rate $m_t$. Let $\Delta_k = t(k-1)-t(k)$; note that $\Delta_k$ follows the exponential distribution of rate $m_{t(k)}$.

Because of the memoryless property of the exponential distribution and independence of edges, $\Delta_k$ for different values of $k$ are independent. By \Cref{lem:gamma-deg-bound}, we have $m_{t(k)} \ge \frac 23 \gamma (k-1)$ for any $k>2^{-2/3}\cdot n$. We couple each $\Delta_k$ with a new random variable $\Delta'_k$, where
$\Delta'_k$ follows an exponential distribution of rate $\frac 23\gamma (k-1)$. Thus, $\Delta_k$ is stochastically dominated by $\Delta'_k$ for any $k>2^{-2/3}\cdot n$.

The random variables $\Delta'_k$ can be coupled with the aforementioned contraction process on a star graph. Let $R$ be a star on $n$ vertices with $n-1$ edges. We use the same notation as above for graph $R$, but with a superscript of $R$. Note that during the contraction process, we always have $m_t^R = n_t^R-1$. 
Therefore, $\Delta^R_k = t^R(k-1)-t^R(k)$ follows an independent exponential distribution of rate $k-1$ by the same argument. We can re-define $\Delta'_k =  \Delta^R_k / (\frac 23 \gamma)$. The coupling now gives that for all $k>2^{-2/3}\cdot n$, we have
\[t(k) = t(k) - t(n) = \sum_{i=k+1}^n \Delta_i \le \sum_{i=k+1}^n \frac{3}{2\gamma}\cdot \Delta^R_i = \frac{3}{2\gamma}\cdot \left(t^R(k)-t^R(n)\right) = \frac{3}{2\gamma}\cdot t^R(k).\]
This implies that either $n_t\le 2^{-2/3} \cdot n$, or $n_t \le n^R_\tau$ where $\tau = \frac{2\gamma}{3}\cdot t$.

Recall that when $t=-\ln q$, $n_t$ is the vertex size of $G(q)$. This corresponds to $\tau = -\frac{2\gamma}{3}\cdot \ln q$ in the contraction process defined for $R$. That is each edge in $R$ fails with probability $e^{-\tau} = q^{\frac{2\gamma}{3}} = 2^{-2/3}$ since $q^{-\gamma} = 2$. Let $X$ be the number of uncontracted edges in $R$ at time $\tau$. Then by Hoeffding's inequality,
\[\Pr[X \ge (1+n^{-0.1})\cdot 2^{-2/3}\cdot (n-1)] \le \exp(-\frac{1}{3}\cdot (1+n^{-0.1})^2 \cdot 2^{-2/3} \cdot (n-1)) = \exp(-\Omega(n^{0.8})).\]
So $n^R_\tau = X+1 \le (2^{-2/3}+n^{-0.1})\cdot n$ whp. By coupling, $|V(H)| = n_{-\ln q}\le (2^{-2/3}+n^{-0.1})\cdot n$ whp.
\end{proof}

\subsection{Bias and Variance of the Estimator}\label{sec:variance-bound}
First, we bound the bias of the estimator. We start by observing that the contraction steps themselves are unbiased. 

\begin{lemma}\label{lem:contract-unbiased}
Assume $H$ is generated by $H\sim G(q)$ for some $q<p$. Then, $u_{H}(p/q)$ is an unbiased estimator of $u_G(p)$.
\end{lemma}
\begin{proof}
The lemma follows from the observation that deleting each edge with probability $p$ from $G$ is equivalent to first creating a set of candidate edges by sampling each edge with probability $q$, and then removing each candidate edge with probability $p/q$.
\end{proof}

We prove a similar property of unbiasedness of sparsification steps (\Cref{lem:sparsify-unbiased} in \Cref{sec:sparsify}). Given these lemmas,  the only source of bias are the base cases. We can now bound overall bias of the estimator as follows:
\begin{lemma}\label{lem:contract-bias}
The algorithm outputs an estimate of $\ugp$ with bias at most 
$0.1\eps\cdot u_G(p)$. 
\end{lemma}
\begin{proof}
Let $X$ be the output of the algorithm.
Consider the base cases of the recursion. Karger's algorithm (\Cref{thm:karger-estimator}) 
for a small instance and the Monte Carlo sampling algorithms have no bias. 

In the last base case of importance sampling, the bias is $O\left(\frac{\log n}{\sqrt{n}}\right) \cdot \ugp$ by \Cref{lem:interface-importance}. We can set the constant threshold for the first base case to be $O(\eps^{-2}\log^2\frac{1}{\eps})$, so that for $n$ larger than the threshold we have $O\left(\frac{\log n}{\sqrt{n}}\right)\le 0.1\eps$.

We use induction on recursion depth to prove that $(1-0.1\eps)\cdot u_G(p) \le \E[X] \le (1+0.1\eps)\cdot u_G(p)$ holds for the recursive case as well.

Consider an inductive step where the algorithm takes the average of two recursive calls. Let $X_1$ and $X_2$ be the output of the two recursive calls, so that $X=\frac{X_1+X_2}{2}$. This can happen for two reasons: recursive contraction or sparsification. First, we consider the case of contraction. Later, we extend the same analysis to sparsification.

Let $H_1$ and $H_2$ be the two recursively generated graphs by contraction. By the inductive hypothesis,
\[|\E[X_i|H_i]-u_{H_i}(p/q)| \le 0.1\eps\cdot u_{H_i}(p/q).\] 

By \Cref{lem:contract-unbiased}, $\E[u_{H_i}(p/q)] = \ugp$. Therefore, 
\begin{align*}
    |\E[X_i]- u_G(p)|
    &= |\E_{H_i}\left[\E[X_i|H_i]\right] - \E_{H_i}[u_{H_i}(p/q)]|
    \le \E_{H_i}[|\E[X_i|H_i]-u_{H_i}(p/q)|]\\
    &\le \E_{H_i}[0.1\eps \cdot u_{H_i}(p/q)] \quad\text{(by inductive hypothesis)}
    = 0.1\eps \cdot \E_{H_i}[u_{H_i}(p/q)]
    = 0.1\eps\cdot u_G(p).
\end{align*}
Thus,
$|\E[X]-u_G(p)| \le \frac 12 \left(|\E[X_1]-u_G(p)|+|\E[X_2]-u_G(p)|\right) \le 0.1\eps \cdot u_G(p)$, using the previous inequality.

In a sparsification step, the algorithm also takes the average of two recursive calls $X=\frac{X_1+X_2}{2}$ on the sparsifier graphs $\tG_1, \tG_2$. To replicate the proof above, we need to show that $u_{\tG_1}(q)$ and $u_{\tG_2}(q)$ are unbiased estimators of $\ugp$. We show this in \Cref{lem:sparsify-unbiased} in \Cref{sec:sparsify}.
\end{proof}

We now need to bound the relative variance of the estimator. Again, there are two cases depending on whether contraction or sparsification is being used in the current computation node. We first prove the bound on the relative variance due to contraction steps below. The bound for sparsification is very similar, and is established formally in \Cref{sec:sparsify}.

Instead of bounding the relative variance of $u_H(p/q)$, we will bound the relative variance of $z_H(p/q)$ and invoke the following lemma:

\begin{lemma}\label{lem:u-from-z}
Assume $p<\theta$. If $z_{H}(q)$ is an unbiased estimator of $\zgp$ with relative variance $\eta$, and $u_{H}(q)$ is an unbiased estimator of $\ugp$, then $u_{H}(q)$ has relative variance at most $\left(1+O\left(\frac{1}{\log n}\right)\right)\eta+O\left(\frac{1}{\log n}\right)$.
\end{lemma}
\begin{proof}
     Since $p<\theta$, we have $u_G(p)\ge \left(1-\frac{1}{\log n}\right) \cdot z_G(p)$ by \Cref{lem:z-approx-u}. We also have $u_{H}(q) \le z_{H}(q)$ for any $q\in [0,1]$. Since $\E[u_{H}(q)] = \ugp$, the relative variance of $u_{H}(q)$ can be bounded as follows:
    \begin{align*}
    \frac{\E[u_{H}(q)^2]}{(\ugp)^2}-1 
    &\le 
    \frac{1}{\left(1-\frac{1}{\log n}\right)^2}\cdot \frac{\E[z_{H}(q)^2]}{(\zgp)^2}-1
    \le \left(1+O\left(\frac{1}{\log n}\right)\right) \eta + O\left(\frac1{\log n}\right) .\qedhere
    \end{align*}
\end{proof}


We now bound the relative variance of $z_H(p/q)$. For the base case, we will use the following known bound:
\begin{lemma}[Lemma 3.1 of \cite{Karger20}]\label{lem:karger-relvar}
Assume $p < \theta$. Sample $H\sim G(q)$ for $q^{-\lambda} = \Theta(1)$. Then, $u_H(p/q)$ is an unbiased estimator of $\ugp$, and $z_H(p/q)$ is an unbiased estimator of $z_G(p)$. Moreover, both estimators have relative second moment upper bounded by
\[q^{-\lambda} \cdot \left(1+O\left(\frac{1}{\log p^{-\lambda}}\right) + O\left(\frac{1}{\log n}\right)\right).\]
\end{lemma}

\begin{lemma}\label{lem:relvar-1step}
Assume that $4n^{-3} \le p^\lambda \le n^{-0.5}$ and $p<\theta$, where $\theta$ is given in \Cref{lem:z-approx-u}. Then the relative second moment of $z_H(p/q)$ is  $2+O\left(\frac{1}{\log n}\right)$ where $H \sim G(q)$ and $q^\gamma = \nicefrac 12$.
\end{lemma}
\begin{proof}
There are two cases depending on the relative values of $\gamma$ and $\lambda$. 

First, suppose $\left(1+\frac{5}{\log n}\right)\gamma \ge \lambda$. 
Then,
\[q^{-\lambda} \le q^{-\gamma(1+ 5/\log n)} =  2^{1+5/\log n} \le 2\left(1+O\left(\frac{1}{\log n}\right)\right) = 2 + O\left(\frac{1}{\log n}\right).\]
By \Cref{lem:karger-relvar},  when $p^\lambda \le n^{-0.5}$, the relative second moment of $z_H(p/q)$ is upper bounded by $q^{-\lambda}\left(1+O\left(\frac{1}{\log n}\right)\right)$, which is at most $2+O\left(\frac{1}{\log n}\right)$ by our bound on $q^{-\lambda}$ above.

In the rest of the proof, we assume $\left(1+\frac{5}{\log n}\right)\gamma < \lambda$. This proof now follows the same template as in Karger's analysis~\cite{Karger20}. The main difference is that we use the parameter $\gamma$ whereas Karger uses $\lambda$. 

The relative second moment is bounded as:
\begin{align*}
\frac{\E[z_H(p/q)^2]}{(\zgp)^2}
&= \frac{1}{(\zgp)^2}\cdot \E\bigg[\sum_{C_i, C_j \in \mathcal{C}(H)}\left(\frac{p}{q}\right)^{c_i+c_j}\bigg] 
= \frac{1}{(\zgp)^2}\sum_{C_i, C_j\in \mathcal{C}(G)} q^{|C_i\cup C_j|}\cdot \left(\frac{p}{q}\right)^{c_i+c_j}\\
&= \frac{1}{(\zgp)^2}\left(\sum_{C_i}\frac{p^{2c_i}}{q^{c_i}} + \sum_{C_i\ne C_j: |C_i\cap C_j|\le \gamma} \frac{p^{c_i+c_j}}{q^{|C_i\cap C_j|}} + \sum_{C_i\ne C_j: |C_i\cap C_j| > \gamma} \frac{p^{c_i+c_j}}{q^{|C_i\cap C_j|}}\right).
\end{align*}
In the above expression, we distinguished between the cases $C_i = C_j$ and $C_i \ne C_j$, and the second case is further split into $|C_i \cap C_j| \le \gamma$ and $|C_i\cap C_j| > \gamma$. We define:
\begin{align*}
    V_1 &= \sum_{C_i}\frac{p^{2c_i}}{q^{c_i}}\\
    V_2 &= \sum_{C_i\ne C_j: |C_i\cap C_j|\le \gamma} \frac{p^{c_i+c_j}}{q^{|C_i\cap C_j|}}\\
    V_3 &= \sum_{C_i\ne C_j: |C_i\cap C_j| > \gamma} \frac{p^{c_i+c_j}}{q^{|C_i\cap C_j|}}.
\end{align*}

We first bound $V_1$:

\[
V_1 - 2\sum_{C_i} p^{2c_i}
= \sum_{C_i} p^{2c_i}\left(\frac{1}{q^{c_i}} - \frac{1}{q^\gamma}\right)
= \frac{p^\gamma}{q^\gamma} \sum_{C_i} p^{c_i} \cdot p^{c_i-\gamma}(q^{-(c_i-\gamma)}-1)
= 2\cdot p^\gamma \sum_{C_i} p^{c_i} \cdot p^{c_i-\gamma}(q^{-(c_i-\gamma)}-1)
\]

Let us also denote $t_i=\frac{c_i-\gamma}{\gamma}$ and $f(t_i) = (2u)^{t_i} - u^{t_i} = p^{c_i-\gamma}(q^{-(c_i-\gamma)}-1)$ where $u=p^\gamma$. Therefore, 
\[
    V_1 - 2\sum_{C_i} p^{2c_i} = 2 p^{\gamma} \sum_{C_i} p^{c_i}\cdot f(t_i).
\]
Note that since $p^{\lambda}\in [n^{-3}, n^{-0.5}]$ and $\gamma\in[\frac 34 \lambda, \lambda]$, we have $u=p^{\gamma}\in [n^{-3}, n^{-0.375}]$. We now use the following fact, which we prove in the appendix:

\begin{fact}\label{fact:dfdt}
For function $f(t) = (2u)^t - u^t$, suppose $u \le n^{-0.375}$ and $t\ge\frac{5}{\log n}$. Then, for $n$ larger than some constant, we have $\frac{\diff f}{\diff t} < 0$.
\end{fact}

Thus, $f$ is monotone decreasing when $t\ge\frac{5}{\log n}$. Since $t_i \ge \frac{\lambda-\gamma}{\gamma}\ge \frac{5}{\log n}$, we can upper bound $f(t_i)\le f\left(\frac{5}{\log n}\right) = u^{5/\log n}\cdot (2^{5/\log n}-1)$. Combined with $p^\gamma \le \zgp$ by \Cref{lem:gamma-z-bound}, we have
\begin{align*}
    V_1 - 2\sum_{C_i}p^{2c_i}
    &\le 2 \cdot p^\gamma \cdot\zgp\cdot f\left(\frac{5}{\log n}\right)
    \le  2\cdot (\zgp)^2 \cdot u^{5/\log n} \cdot (2^{5 / \log n}-1)\\
    &\le 2(\zgp)^2 \left(n^{-0.375}\right)^{5/\log n} \cdot (2^{5/\log n}-1)
    \le O\left(\frac{1}{\log n}\right)\cdot (\zgp)^2.
\end{align*}
Thus,
\begin{align}\label{eq:v1}
    V_1 \le O\left(\frac{1}{\log n}\right)\cdot (\zgp)^2 + 2\sum_{C_i}p^{2c_i}.
\end{align}

Next, we bound $V_2$ as follows:
\begin{align}\label{eq:v2}
V_2 
= \sum_{C_i\ne C_j, |C_i\cap C_j|\le \gamma} \frac{p^{c_i+c_j}}{q^{|C_i\cap C_j|}}
\le q^{-\gamma}\cdot \sum_{C_i, C_j} p^{c_i+c_j} - q^{-\gamma}\cdot \sum_{C_i} p^{2c_i}
= 2(\zgp)^2 - 2\cdot \sum_{C_i} p^{2c_i}.
\end{align}

Finally, we bound $V_3$ as follows:
\[
V_3 
= \sum_{C_i\ne C_j, |C_i\cap C_j|>\gamma} p^{|C_i\cup C_j|}\left(\frac{p}{q}\right)^{|C_i\cap C_j|}
\le \sum_{C_i\ne C_j} p^{|C_i\cup C_j|}\left(\frac{p}{q}\right)^{\gamma}
= q^{-\gamma} \cdot x_G(p)\cdot p^{\gamma} 
\le 2\cdot x_G(p) \cdot \zgp.
\]
Since $p< \theta$, \Cref{lem:z-approx-u} gives $x_G(p) \le \frac{1}{\log n}\cdot \zgp$. Therefore,
\begin{align}\label{eq:v3}
    V_3 \le  2\cdot x_G(p) \cdot \zgp \le \frac{2}{\log n}\cdot (\zgp)^2.
\end{align}    

Putting the bounds on $V_1, V_2$ and $V_3$ given by \Cref{eq:v1,eq:v2,eq:v3} together, we get
\[
\frac{\E[z_H(p/q)^2]}{(\zgp)^2} 
\le \frac{V_1+V_2+V_3}{(\zgp)^2}  
= 2 + O\left(\frac{1}{\log n}\right).\qedhere
\]
\end{proof}

Combining \Cref{lem:relvar-1step,lem:karger-relvar,lem:u-from-z}, we obtain the following:
\begin{corollary}\label{cor:contract-u-var}
$u_H(p/q)$ is an unbiased estimator of $u_G(p)$ with relative second moment at most $2+O\left(\frac{1}{\log n}\right)$.
\end{corollary}

We now use induction to bound the second moment of the overall estimator. This requires our bounds on the base cases as well as that established for a recursive contraction step in \Cref{lem:relvar-1step} and the corresponding bound for a sparsification step that we establish in \Cref{lem:sparsify-before-contract} in \Cref{sec:sparsify}.

\begin{lemma}\label{lem:contract-var}
    The second moment of the estimator given by the overall algorithm is at most $\log^{O(1)} n \cdot (\ugp)^2$ whp. 
\end{lemma}
\begin{proof}
Let $X$ be the estimator given by the overall algorithm.
We use induction on recursion depth to prove $\E[X^2] \le (\log^K n) \cdot (u_G(p))^2$ for some constant $K\ge 1$.

Consider the base cases of the recursion.
In the first case (Karger's algorithm (\Cref{thm:karger-estimator}) 
the relative variance is $O(1)$.
In the second case (Monte Carlo sampling), we get an unbiased estimator of $u_G(p)$ with relative variance $O(1)$ by \Cref{lem:mc,lem:mc-2step}.
In the last case (importance sampling on spanning tree packing), we get an estimator of $\ugp$ with relative bias $0.1\eps$ (by \Cref{lem:contract-bias}) and relative variance $O(1)$ (by \Cref{lem:interface-importance}).
For such an estimator $X'$, we have $\E[X']\le (1+0.1\eps)u$ and $\E[X'^2]\le O(1)\cdot \E[X']^2\le O(1)\cdot (1+0.1\eps)^2(\ugp)^2= O(1)\cdot (\ugp)^2$.

Therefore, the statement of the lemma holds for all the base cases.

Next we consider the inductive step where the algorithm takes average of two recursive calls. Let $X_1$ and $X_2$ be the outputs of the two recursive calls, so that $X=\frac{X_1+X_2}{2}$. Again, we have two cases depending on whether we are in a contraction step or a sparsification step. First, we consider a contraction step. The sparsification step is similar and handled at the end. 

Let $H_1$ and $H_2$ be the two random contracted graphs. 
%
%
Now, we have
\[\E[X^2] = \E\left[\left(\frac{X_1+X_2}{2}\right)^2\right] = \E_{H_1,H_2}\left[\frac 14 (\E[X_1^2|H_1] + \E[X_2^2|H_2]) + \frac 12 \cdot \E[X_1|H_1] \cdot \E[X_2|H_2]\right],\]
since $X_1, X_2$ are respectively independent of $H_2, H_1$.

By the inductive hypothesis, $\E[X_i^2|H_i] \le (\log^K n_i) \cdot (u_{H_i}(p/q))^2$ for $i=1,2$. By applying \Cref{lem:contract-bias} on the recursive calls, we have 
\[
\E[X_i|H_i] \le \left(1+0.1\eps\right)\cdot u_{H_i}(p/q).
\]
We can now bound the second moment by
\begin{align}
\E[X^2]
&\le  \frac 14\cdot \E_{H_1, H_2}\left[ (\log^K n_1)\cdot (u_{H_1}(p/q))^2 + (\log^K n_2) \cdot (u_{H_2}(p/q))^2\right] \nonumber\\
& \qquad \qquad + \frac 12 \left(1+0.1\eps\right)^2 \cdot \E_{H_1, H_2}[u_{H_1}(p/q)\cdot u_{H_2}(p/q)]\nonumber\\
&\le  \frac 14 \cdot (\log^K n_1) \cdot \E_{H_1}[(u_{H_1}(p/q))^2] + \frac 14 \cdot (\log^K n_2) \cdot \E_{H_2}[(u_{H_2}(p/q))^2] \nonumber\\
& \qquad \qquad +  \E_{H_1}[u_{H_1}(p/q)]\cdot \E_{H_2}[u_{H_2}(p/q)],\label{eq:xsq}
\end{align}
by independence of $H_1, H_2$ and since $\eps<1$. 

To bound the first term in \Cref{eq:xsq}, note that by \Cref{cor:contract-u-var}, we have 
\[
    \E[(u_{H_i}(p/q))^2] \le \left(2+O\left(\frac{1}{\log n}\right)\right)(u_G(p))^2 \text{ for }i=1, 2.
\]

To bound the second term in \Cref{eq:xsq}, note that $\E[u_{H_i}(p/q)] = u_G(p)$ by \Cref{lem:contract-unbiased}. Thus, 
\[
    E_{H_1}[u_{H_1}(p/q)]\cdot \E_{H_2}[u_{H_2}(p/q)]
    = (\ugp)^2.
\]


Putting the two terms together, we get
\[\E[X^2]\le \frac 14 (\log^K n_1 + \log^K n_2)\left(2+O\left(\frac{1}{\log n}\right)\right) (u_G(p))^2
+  (u_G(p))^2.\]

By \Cref{lem:contract-size-dec}, $n_i\le (2^{-2/3}+n^{-0.1})\cdot n$ whp. Then, for a fixed constant $c \in (0, 1)$, we have $\log^K n_i\le (\log n-c)^K$ for $i=1,2$. Continuing the bound on $\E[X^2]$,
\begin{align*}
\E[X^2] 
&\le \frac 12\cdot (\log n-c)^K\left(2+O\left(\frac{1}{\log n}\right)\right) (u_G(p))^2 + (u_G(p))^2\\
\frac{\E[X^2] }{(\ugp)^2}
&\le (\log n-c)^K \left(1+O\left(\frac{1}{\log n}\right)\right) + 1
\le \log^K n \cdot \left(1-\frac{c}{\log n}\right)^K \cdot \left(1+\frac{c}{\log n}\right)^{O(1)}\le \log^K n,
\end{align*}
for a large enough $K$ and $n$ larger than some constant.

For a sparsification step, we can repeat the same argument after replacing $H$ by $\tG$ and $u_H(p/q)$ by $u_{\tG}(q)$. By \Cref{lem:sparsify-unbiased}, $u_{\tG}(q)$ is also an unbiased estimator of $\ugp$. The bound for relative second moment of $u_{\tG}(q)$ is $2+O\left(\frac{1}{\log n}\right)$ by \Cref{cor:sparsify-u-var}, which is as good as \Cref{cor:contract-u-var}.
\end{proof}


\subsection{Running Time of the Algorithm}

We first establish the following property enforced by sparsification steps.

\begin{lemma}\label{lem:runtime-sparsify-step}
Assume $m > n^{1.5-\xi}$ and the algorithm executes a sparsification step followed by a contraction step. Let $m'$ be the number of edges in a resulting graph. Then, the following bounds hold in expectation: $m'= \tO(n)$ and $\log^2 m' \le 2^{-2/3} \cdot \log^2 m$.
\end{lemma}
\begin{proof}
Sparsification generates a graph $\tG$ with min-cut value $\tilde{\lambda}=\tO(1)$. Then we perform random contraction on $\tG$ with $q^\gamma = \nicefrac 12$, where $\gamma \le \tilde{\lambda}$. Then $q^{\tilde{\lambda}}\le q^{\gamma} = \nicefrac 12$ and $\frac{1}{1-q} = O(\tilde{\lambda})$. By \Cref{lem:contraction-size-bound}, the expected edge size of the graph after contraction (which follows distribution $\tG(q)$) is at most $\frac{n}{1-q}= O(n\tilde{\lambda}) = \tO(n)$. (Note that $n$ is the vertex size before sparsification.)

The assumption gives $m > n^{1.5-\xi}$. After contraction we have $\E[m'] \le \tO(n) \le \tO(m^{2/3+\xi})$. Notice that $\log^2(x)$ is concave when $x>e$ and we can apply Jensen's inequality. For $\xi < 0.1$,
\[
\E[\log^2 m']
\le \log^2 \E[m']
\le \log^2 \tO(m^{2/3+\xi})
\le \left(\left(\frac 23 + \xi + o(1)\right)\log m\right)^2
\le 2^{-2/3}\cdot \log^2 m. \qedhere
\]
\end{proof}

\begin{lemma}\label{lem:contraction-runtime}
The algorithm runs in $m^{1+o(1)}+\tO(n^{1.5}\eps^{-1})$ time in expectation.
\end{lemma}
\begin{proof}
Let $T(n, m)$ be the expected running time of a recursive call on a graph with $n$ vertices and $m$ edges. Recall that there are three types of nodes in the computation tree: base cases, contraction nodes, and sparsification nodes. If the parent of a contraction node is a sparsification node, we call it an {\em irregular} contraction node; otherwise, the contraction node is called a {\em regular} contraction node. First, we shortcut all irregular contraction nodes (by making their parent the parent of their children). The running time of all such irregular contraction nodes are accounted for by their parent sparsification nodes. Note that each sparsification node has at most two irregular contraction nodes as children; hence, it accounts for the cost of at most three nodes (itself and its two children irregular contraction nodes). Because of this transformation, the number of children of a sparsification node can increase to at most 4.

In the rest of the discussion, we assume that the contraction tree has only three types of nodes: sparsification nodes, regular contraction nodes and base cases (which are leaves of the computation tree). 
In the first base case of $n<O(\eps^{-2}\log^2\frac{1}{\eps})$, the running time is $\tO(n^2)=\tO(n^{1.5}\eps^{-1})$ by \Cref{thm:karger-estimator}.
The second and third base cases take $m^{1+o(1)}+\tO(n^{1.5})$ time by \Cref{lem:interface-importance,lem:mc,lem:mc-2step}.
In a sparsification node, the sparsification algorithm in \Cref{sec:sparsify} takes $O(m)$ time. 
In a contraction node, constructing a Gomory-Hu tree for unweighted graphs takes time $m^{1+o(1)}$ \cite{Abboud22}, which dominates all other operations. 
So the time spent at any node of the computation tree is $m^{1+o(1)}+\tO(n^{1.5}\eps^{-1})$ in total (including the charge received by a sparsification node from its children irregular contraction nodes).

\Cref{lem:contract-size-dec} shows that whp each recursive contraction reduces the vertex size by a factor of $2^{-2/3}+O(n^{-0.1})$. In particular, this holds for a regular contraction node and its children, as well as a sparsification node and its children inherited from an irregular contraction child. 

First, we consider regular contraction nodes. Note that these nodes still have at most two children.
 Moreover, they satisfy $m\le n^{1.5-\xi}$ (which implies $m^{1+o(1)} \le \tO(n^{1.5})$). Therefore, the recurrence is
\begin{align}\label{eq:recur-contract}
    T(n, m) \le \tO(n^{1.5}\eps^{-1}) + 2 T(2^{-2/3}(1+n^{-0.1})\cdot n, m).
\end{align}    

Now, we consider a non-root sparsification node. This is more complicated because $m > n^{1.5-\xi}$. Recall that the running time incurred at this node (including that inherited from irregular contraction children) is $m^{1+o(1)}+\tO(n^{1.5}\eps^{-1})$. Of these, the term $\tO(n^{1.5}\eps^{-1})$ can be handled as in the previous case, namely it appears in the recursion. The other term $m^{1+o(1)}$ is charged to the parent of the sparsification node. Let $\hat{n}, \hat{m}$ respectively represent the number of vertices and edges in the parent node. If the parent is a regular contraction node, then we have $m \le \hat{m} \le \hat{n}^{1.5-\xi}$. In this case, we can charge the $m^{1+o(1)}$ term to the parent's recurence relation \Cref{eq:recur-contract}. Otherwise, the parent is a sparsification node and we have $m=\tO(\hat{n})$ in expectation by \Cref{lem:runtime-sparsify-step}. So, we can also charge $m^{1+o(1)}$ to $\tO(\hat{n}^{1.5})$. Finally, note that a sparsification node has at most 4 children. Let $m'$ denote the number of edges in any child of the sparsification node. We can write the following recurrence for a sparsification node:
\begin{align}\label{eq:recur-sparsify}
    T(n, m) \le \tO(n^{1.5}\eps^{-1}) + 4 T(2^{-2/3}(1+n^{-0.1})\cdot n, m'), \quad \text{where $m'$ satisfies \Cref{lem:runtime-sparsify-step}.}
\end{align}

For the sake of the master theorem, we define a potential $\rho = n\cdot (1+n^{-0.1})\cdot \log^2 m$. \Cref{lem:runtime-sparsify-step} measures the progress in $\log^2 m$ for the second recurrence (\Cref{eq:recur-sparsify}) by $\log^2 m' \le 2^{-2/3} \cdot \log^2 m$. We have
\[T(\rho) = \tO(\rho^{1.5}\eps^{-1}) + 2T(2^{-2/3}\rho) \text{ or } T(\rho) = \tO(\rho^{1.5}\eps^{-1}) + 4T(2^{-4/3}\rho)\]
The solution is $T(\rho) = \tO(\rho^{1.5}\eps^{-1})$, or $T(n, m) = \tO(n^{1.5}\eps^{-1})$.

Finally, note that the root node takes $m^{1+o(1)}+\tO(n^{1.5}\eps^{-1})$ time (and unlike non-root nodes, the $m^{1+o(1)}$ term isn't chargeable elsewhere). Therefore, the overall running time is $m^{1+o(1)} + \tO(n^{1.5}\eps^{-1})$ in expectation.
\end{proof}
    \Cref{lem:interface-contract} now follows from \Cref{lem:contract-bias,lem:contract-var,lem:contraction-runtime}. \qed

\section{Sparsification in Recursive Contraction}
\label{sec:sparsify}
This section describes the sparsification step, which is used in recursive contraction at the root of the computation tree, or when the edge size $m > n^{1.5-\xi}$. The goal is to reduce the minimum cut value to $O(\log^3 n)=\tO(1)$.


We apply the sparsification lemma (\Cref{lem:sparsify}) with parameter $\delta=O\left(\frac{1}{\log n}\right)$ to obtain a sparsifier graph $\tilde{G}$. The graph $\tilde{G}$ is generated by picking each edge independently with probability $\alpha= \Theta\left(\frac{\log^3 n}{\lambda}\right)$ from $G$. $\tilde{G}$ has min-cut value $\tilde{\lambda}=O(\log^3 n)$, which is our desired property.

In the rest of this section, we show that there is a value $q$ such that $u_{\tilde{G}}(q)$ is an unbiased estimator of $\ugp$ with $O(1)$ relative variance. This would allow us to focus on estimating $u_{\tilde{G}}(q)$ in the recursive contraction algorithm.

Choose $q$ such that $1-q =\frac{1-p}{\alpha}$.
Note that $q < p$ by the following argument: since $\lambda>\Omega(\log^3n)$, we have $\alpha < 1$, which implies $1-q > 1-p$ and $q < p$. 

\begin{lemma}\label{lem:sparsify-unbiased}
    Assume $\tG$ is generated by picking every edge in $G$ independently with probability $\alpha$. Also, suppose $q$ is chosen such that $1-q = \frac{1-p}{\alpha}$. Then, $u_{\tilde{G}}(q)$ is an unbiased estimator of $u_G(p)$.
\end{lemma}
\begin{proof}
Keeping each edge with probability $1-p$ is equivalent to first choosing it with probability $\alpha$ and then keeping it with probability $1-q = \frac{1-p}{\alpha}$.
\end{proof} 

Our goal now is to bound the relative variance of $u_{\tilde{G}}(q)$. Instead, we will bound the relative variance of $z_{\tilde{G}}(q)$ and invoke \Cref{lem:u-from-z}.
%
We now bound the relative second moment of $z_{\tG}(q)$:
\begin{lemma} \label{lem:sparsify-before-contract}
    When $n^{-3} \le p^{\lambda} \le n^{-0.5}$, we have that $z_{\tilde{G}}(q)$ is an unbiased estimator of $\zgp$ with relative second moment at most $2+O\left(\frac1{\log n}\right)$.
\end{lemma}
\begin{proof}
Since $p^\lambda \ge n^{-3}$, we have $1-p \le 1-e^{-3\ln n/\lambda}\le O\left(\frac{\log n}{\lambda}\right)$. Similarly, $p^\lambda \le n^{-0.5}$ implies that $1-p \ge 1-e^{-0.5\ln n/\lambda} \ge \Omega\left(\frac{\log n}\lambda\right)$. Therefore, denote 
\begin{equation*}
\tau = 1-p = \Theta\left(\frac{\log n}\lambda\right).
\end{equation*}
Recall from above that we defined $q$ so as to satisfy 
\begin{equation*}
1-q = \frac{\tau}{\alpha} = \Theta\left(\frac{\frac{\log n}{\lambda}}{\frac{\log^3 n}{\lambda}}\right) = \Theta\left(\frac{1}{\log^2 n}\right).
\end{equation*}

Let $Y_e$ be the indicator that edge $e$ is picked by the random graph ${\tilde{G}}$.  For any edge $e$, we have
\begin{align}
\E\left[q^{Y_e}\right] &= \alpha q + (1-\alpha) = 1-\alpha(1-q) = p \label{eq:qtoY} \\ 
\E\left[q^{2Y_e}\right] &= \alpha q^2 + (1-\alpha) = 1-\tau (1+ q) \le (1-\tau)(1-\tau q) = p\cdot (1-\tau q).\label{eq:qto2Y}
\end{align}
 We can bound $\E\left[q^{2Y_e}\right]$ in two ways:
\begin{align}
\frac{\E\left[q^{2Y_e}\right]}{p} &\le  1-\tau q \label{eq:qto2Yfirst}\\
\frac{\E\left[q^{2Y_e}\right]}{p^2} &\le  \frac{1-\tau q}{p} = \frac{(1-\tau)+\tau(1-q)}{1-\tau} = 1+\frac{\tau(1-q)}{1-\tau} \le 1+2\tau(1-q) = 1+O\left(\frac{1}{\lambda\log n}\right). \label{eq:qto2Ysecond}
\end{align}

Next we calculate the expectation and relative variance of $z_{\tilde{G}}(q)$. Notice that $Y_e$'s are independent for each edge $e$. Use $C_i\Delta C_j$ to denote the symmetric difference $(C_i\setminus C_j)\cup (C_j\setminus C_i)$ over two cuts $C_i, C_j$. Use $\tilde{d}(\cdot)$ to denote the cut value function in $\tilde{G}$.  First, we calculate the expectation of $z_{\tilde{G}}(q)$:
\[
\E[z_{\tilde{G}}(q)]
= \E\left[\sum_{C_i} q^{\tilde{d}(C_i)}\right] 
= \sum_{C_i}\E\left[q^{\sum_{e\in C_i}Y_e}\right]
= \sum_{C_i}\prod_{e\in C_i} \E\left[q^{Y_e}\right] 
\stackrel{(\ref{eq:qtoY})}{=} \sum_{C_i}p^{c_i}  
= z_G(p).
\]
Next, we bound the second moment of $z_{\tilde{G}}(q)$
\begin{align*}
    \E\left[(z_{\tilde{G}}(q))^2\right] 
    &= \E\left[\sum_{C_i}\sum_{C_j}q^{\tilde{d}(C_i)+\tilde{d}(C_j)}\right] 
    =\sum_{C_i}\sum_{C_j} \E\left[q^{\sum_{e\in C_i}Y_e + \sum_{e\in C_j}Y_e}\right] \\
    &=\sum_{C_i}\sum_{C_j}\prod_{e\in C_i\cap C_j}\E\left[q^{2Y_e}\right] \prod_{e\in C_i\Delta C_j} \E\left[q^{Y_e}\right]
    \stackrel{(\ref{eq:qtoY})}{=} \sum_{C_i}\sum_{C_j}p^{|C_i\Delta C_j|} \cdot \left(\E\left[q^{2Y_e}\right]\right)^{|C_i\cap C_j|}.
\end{align*}
We partition this sum into three parts and separately bound their ratios with $(\zgp)^2$.

For terms with $C_i=C_j$,
\begin{align}
&\frac{\sum_{C_i} p^{|C_i\Delta C_i|} \cdot \left(\E\left[q^{2Y_e}\right]\right)^{|C_i\cap C_i|}}{(\zgp)^2} 
= \frac{\sum_{C_i} \left(\E\left[q^{2Y_e}\right]\right)^{c_i}}{(\zgp)^2} 
= \frac{\sum_{C_i} p^{c_i}\cdot \left(\frac{\E\left[q^{2Y_e}\right]}{p}\right)^{c_i}}{\sum_{C_i} p^{c_i} \cdot \zgp} 
\le \max_{C_i} \frac{\left(\frac{\E\left[q^{2Y_e}\right]}{p}\right)^{c_i}}{\zgp} \nonumber \\
&\stackrel{(\ref{eq:qto2Yfirst})}{\le} \frac{(1-\tau q)^\lambda}{\zgp} 
\le \left(\frac{1-\tau q}{p}\right)^\lambda
\stackrel{(\ref{eq:qto2Ysecond})}{\le} \left(1+O\left(\frac{1}{\lambda\log n}\right)\right)^\lambda = 1+O\left(\frac{1}{\log n}\right). \label{eq:equal-terms-summation}
\end{align}

For terms with $|C_i\cap C_j|\le \lambda$,
\begin{align*}
&\frac{\sum_{|C_i\cap C_j|\le\lambda} p^{|C_i\Delta C_j|}\cdot \left(\E\left[q^{2Y_e}\right]\right)^{|C_i\cap C_j|}}{(\zgp)^2}
=\frac{\sum_{|C_i\cap C_j|\le \lambda}\, p^{c_i+c_j}\cdot \left(\frac{\E\left[q^{2Y_e}\right]}{p^2}\right)^{|C_i\cap C_j|}}{(\zgp)^2} \\
&\stackrel{(\ref{eq:qto2Ysecond})}{\le} \frac{\sum_{|C_i\cap C_j|\le \lambda}\, p^{c_i+c_j}\left(1+O\left(\frac1{\lambda\log n}\right)\right)^{\lambda}}{\sum_{C_i,C_j}p^{c_i+c_j}}  
\le \left(1+O\left(\frac{1}{\lambda\log n}\right)\right)^\lambda = 1+O\left(\frac{1}{\log n}\right).
\end{align*}

For terms with $|C_i\cap C_j| > \lambda$ and $C_i\ne C_j$, we have 
\begin{align*}
&\frac{\sum_{C_i\ne C_j,|C_i\cap C_j|>\lambda} \, p^{|C_i\Delta C_j|}\cdot \E\left[q^{2Y_e}\right]^{|C_i\cap C_j|}}{(\zgp)^2} 
= \frac{\sum_{C_i\ne C_j,|C_i\cap C_j|>\lambda} \, p^{|C_i\cup C_j|} \cdot \left(\frac{\E\left[q^{2Y_e}\right]}{p}\right)^{|C_i\cap C_j|}}{(\zgp)^2} \\
&\stackrel{(\ref{eq:qto2Yfirst})}{\le} \frac{\sum_{C_i\ne C_j,|C_i\cap C_j|>\lambda} \, p^{|C_i\cup C_j|}\cdot (1-\tau q)^{\lambda}}{(\zgp)^2} 
\le \frac{x_G(p) \cdot (1-\tau q)^{\lambda}}{(\zgp)^2} \quad \text{(by definition of $\xgp$)}.
\end{align*}
Applying $\frac{\xgp}{\zgp} \le \frac1{\log n}$ from \Cref{lem:z-approx-u}, this is at most
\[\frac{1}{\log n}\cdot \frac{(1-\tau q)^\lambda}{\zgp} 
\stackrel{(\ref{eq:equal-terms-summation})}{\le} \frac{1}{\log n}\cdot \left(1+O\left(\frac{1}{\log n}\right)\right) 
\le O\left(\frac1{\log n}\right) .\]
In conclusion, the total relative variance is given by
\[\frac{\E[z_{\tilde{G}}(q)^2]}{(\zgp)^2} \le \left( 1+O\left(\frac1{\log n}\right)\right)+\left( 1+O\left(\frac1{\log n}\right)\right)+O\left(\frac1{\log n}\right)=2 + O\left(\frac1{\log n}\right).\qedhere \]
\end{proof}

Combining \Cref{lem:sparsify-before-contract,lem:u-from-z}, we obtain the following:
\begin{corollary}\label{cor:sparsify-u-var}
$u_{\tilde{G}}(q)$ is an unbiased estimator of $u_G(p)$ with relative second moment at most $2+O\left(\frac{1}{\log n}\right)$.
\end{corollary}


\section{Monte Carlo Sampling}
\label{sec:mc}
Finally, we use Monte Carlo sampling for the unreliable case. 
We use it in two different ways: na\"ive Monte Carlo sampling and two-step recursive Monte Carlo sampling. These two algorithms respectively handle the cases $p \ge \theta$ and $n^{-\frac{1}{2\lambda}} < p < \theta$, for $u_G(\theta) = n^{-O(1/\log \log n)}$. (The precise value of $\theta$ is the one given in \Cref{lem:z-approx-u}.) 
Note that in conjunction with the previous sections, this covers all possibilities.

\subsection{Na\"ive Monte Carlo Sampling}\label{sec:naive}
When $p \ge \theta$, we run a na\"ive Monte Carlo sampling algorithm. 
%
%
Our goal is to show \Cref{lem:mc}, which we restate below:

\mc*

In each round of this algorithm, we run the following sampling process: remove each edge independently with probability $p$ and check whether the graph gets disconnected. The corresponding indicator variable $X$ is a Bernoulli random variable with parameter $\ugp$. The expectation, variance, and relative variance of this variable are given below:
\begin{lemma}\label{lem:mc-single}
For a single round of Monte Carlo sampling, the mean, variance, and relative variance of the estimator are given by
$\ex[X] = \ugp$, $\var[X] = \ugp(1-\ugp)$ and $\relv[X] = \frac{\ugp(1-\ugp)}{(\ugp)^2} \le \frac{1}{\ugp}$. 
\end{lemma}

By repeated Monte Carlo sampling for $\tO(1/\ugp)$ independent rounds, we can reduce the relative variance to $O(1)$ by \Cref{fact:rel-var-decrease}. The running time of each round is $O(m)$, so the total running time is $\tO(m/\ugp)$. Note that $\ugp \ge u_G(\theta) = n^{-o(1)}$; therefore, the running time is $m^{1+o(1)}$. 

This completes the proof of \Cref{lem:mc}.

\eat{
Then, we can get a $(1\pm\eps)$-approximation of $\ugp$ whp by applying \Cref{lem:mc-sample}. This yields the following lemma:

\begin{lemma}\label{lem:mc}
    For any $p$ such that $\ugp > n^{-o(1)}$, there is an algorithm that obtains a $(1\pm\eps)$-approximation to $\ugp$ with probability $\ge 1-\frac{1}{\poly(n)}$ and runs in time $m^{1+o(1)}$.
\end{lemma}
}

\subsection{Two-step Monte Carlo Sampling}\label{sec:two-step}
When $p < \theta$ and $p > n^{-\frac{1}{2\lambda}}$, we use a two-step Monte Carlo sampling algorithm. This follows a sparse sampling technique used by Karger~\cite{Karger16}. Our goal is to show \Cref{lem:mc-2step}, which we restate below:

\mctwo*


To describe this algorithm, we first note that instead of removing each edge with probability $p$ and checking if the graph is disconnected, we can equivalently {\em contract} every edge with probability $1-p$ and check if we get more than one vertex. Instead of doing this in one shot, we stage this contraction process out into two steps: in the first step, for some $q > p$, we contract each edge with probability $1-q$ to form a graph $H$ (i.e., $H\sim G(q)$), and then in the second step, we contract each edge in $H$ with probability $1-p/q$. Note that the indicator variable for obtaining $> 1$ vertex at the end of this two step contraction process is an unbiased estimator for $u_H(p/q)$, and since $\ex[u_H(p/q)] = \ugp$, is also an unbiased estimator of $\ugp$.

But, what do we gain in this two-step process? To understand this, we need to bound the running time for each of the two steps. In each step, we bound the running time for a single round of Monte Carlo sampling and also the relative variance of the resulting estimator, which in turn bounds the number of rounds by \Cref{lem:mc}.

For the first step, we use a na\"ive implementation of Monte Carlo sampling in $O(m)$ time. To ensure efficiency in terms of the number of rounds of sampling, we need to choose $q$ to be large enough such that the relative variance at the end of this step is small. We choose $q$ such that $q^\lambda = \nicefrac 12$. Note that if $q^\lambda = \nicefrac 12$ and $p < \theta$  (for $\theta$ given in \Cref{lem:z-approx-u}), then 
\[
q^\lambda = \frac 12 > n^{-O(1/\log\log n)} = u_G(\theta) \ge \ugp \ge p^\lambda;
\]
 thus $q > p$ as required. Furthermore, since $p \le \theta$, we can apply \Cref{lem:karger-relvar}.
Note that $p^\lambda \le \ugp \le u_G(\theta) = n^{-O(1/\log \log n)}$ (by \Cref{lem:z-approx-u}). Thus, 
\[
O\left(\frac{1}{\log p^{-\lambda}}\right) = O\left(\frac{\log\log n}{\log n}\right) = o(1).
\]
Since $q^{-\lambda} = 2$, \Cref{lem:karger-relvar} implies that the relative variance of $u_H(p/q)$ over the randomness of $H$ is at most $2 + o(1) = O(1)$.

Now, we consider the second step of Monte Carlo sampling. Here, we na\"ively bound the relative variance of the estimator by the relative variance of the overall estimator, which by \Cref{lem:mc-single} is given by $1/\ugp$. Now, since $\ugp \ge p^\lambda > n^{-1/2}$, we get that the relative variance of the estimator in the second step is at most $\sqrt{n}$. Since we are aiming for a running time of $\tO(m+n^{3/2})$, we must give an implementation of the second step of Monte Carlo sampling in $\tO(n)$ time. Crucially, because of edge contractions, the number of edges in $H$ generated after the first step of sampling is at most $\frac{n}{1-q}$ in expectation (by \Cref{lem:contraction-size-bound}).
By our choice of $q$ that ensures $q^\lambda = \nicefrac 12$, \Cref{lem:contraction-size-bound} implies that in expectation, $H$ contains $O(n\lambda)$ edges. So, a na\"ive implementation of Monte Carlo sampling takes $O(n\lambda)$ time per round in expectation. But, this is still not enough since we are aiming for a running time of $\tO(n)$ per round. Recall that the probability of contracting an edge in $H$ is $1 - p/q$, which can be bounded as follows (the second inequality uses $p^\lambda > n^{-1/2}$):
\[1-\frac{p}{q} < 1-p \le 1-e^{-0.5\ln n/\lambda} \le \frac{0.5\ln n}{\lambda} = O\left(\frac{\log n}{\lambda}\right).\]
Now, since $H$ only has $O(n\lambda)$ edges in expectation, it follows that the expected number of contracted edges in $H$ is $O(n\log n)$. 

Instead of iterating over all edges in $H$, we directly choose the edges to contract and check if they contain a spanning tree over the set of vertices in $H$. The first step is to determine the number of edges to contract. Note that this is a Binomial random variable with parameters $|E(H)|$ and $1-p/q$. Once we have selected the number of contracted edges by generating the Binomial random variable, we must then select these edges to contract uniformly at random from the edges in $H$. We choose edges uniformly at random (with replacement, discarding if we get a duplicate) until we have the desired number of edges. This algorithm requires $\tO(1)$ time per contracted edge, which is $\tO(n)$ time in expectation overall. 

We now describe our overall algorithm. For each round of the first sampling step that reduces $G$ to $H$, we use $O(\sqrt{n})$ independent rounds of the second step of Monte Carlo sampling on $H$. This gives an estimator of $u_H(p/q)$ with relative variance 1, and takes $\tO(m+n^{1.5})$ time. The overall relative variance of the estimator of $u_G(p)$ is now $O(1)$ by \Cref{lem:relvar-multiply}. So, we now invoke \Cref{lem:mc-sample} by repeating this sampling for $\tO(\eps^{-2})$ rounds to get a $(1\pm\eps)$-approximate estimate of $\ugp$ in $\tO((m+n^{1.5})\eps^{-2})$ time overall.

This completes the proof of \Cref{lem:mc-2step}.

\section{Conclusion}
\label{sec:conclusion}

We obtain an algorithm for the network unreliability problem that runs in $m^{1+o(1)} + \tO(n^{1.5})$ time and improves on the previous best running time of $\tO(n^2)$. Our main technical contribution is a new algorithm for estimating unreliability in the {\em very reliable} situation, which is normally the bottleneck for unreliability algorithms. Our algorithm utilizes a carefully defined importance sampling procedure on a collection of cuts defined via a spanning tree packing of the graph. In addition, for the {\em moderately reliable} setting, we give a new, improved analysis for (a version of) the recursive contraction algorithm. Our algorithm is almost-linear time for dense instances ($m = \Omega(n^
{1.5})$). Obtaining an almost linear running time for all instances is the natural eventual goal.  But, this will require new ideas that go beyond the techniques described in this paper which are optimized to obtain the $\tO(n^{1.5})$ bound.

\eat{
We set our threshold of moderate and very reliable cases to be $p^\lambda \approx n^{-3}$. This, of course, is decided by the tradeoff between running times of the two cases. Our techniques can achieve almost-linear running time under stronger assumptions, which is roughly $p^\lambda< n^{-4}$ for importance sampling and $p^\lambda > n^{-1}$ for recursive contraction. 
}

Our result holds for the {\em uniform} case where every edge fails with the same probability $p$. A more general setting is when each edge $e$ fails with a different probability $p_e$. This can be simulated in the uniform setting by using multiple parallel edges for each edge, and our algorithm works for this simulated graph but at the cost of an increase in the running time because of the dependence on the number of edges $m$.

The problem of estimating reliability, i.e., the probability that the graph {\em stays connected} under edge failures, is equivalent to the unreliability problem when we seek exact computation. However, for approximation algorithms, the problems are different. For approximating reliability, the bottleneck graphs are the very unreliable ones, and require an entirely different set of techniques. For the reliability problem, the current best running time is $\tO\left(\frac{mn^2}{\eps^2(1-p)}\right)$ due to Guo and He~\cite{GuoH20}. Improving this bound is also an interesting question in the general area of network reliability.

\bibliographystyle{plain}
\bibliography{refs}

\begin{thebibliography}{10}

\bibitem{Abboud22}
Amir Abboud, Robert Krauthgamer, Jason Li, Debmalya Panigrahi, Thatchaphol
  Saranurak, and Ohad Trabelsi.
\newblock Breaking the cubic barrier for all-pairs max-flow: Gomory-hu tree in
  nearly quadratic time.
\newblock In {\em 2022 IEEE 63rd Annual Symposium on Foundations of Computer
  Science (FOCS)}, pages 884--895, 2022.

\bibitem{AlonFW95}
Noga Alon, Alan~M. Frieze, and Dominic Welsh.
\newblock Polynomial time randomized approximation schemes for
  tutte-gr{\"{o}}thendieck invariants: The dense case.
\newblock {\em Random Struct. Algorithms}, 6(4):459--478, 1995.

\bibitem{chaturvedi2016network}
Sanjay~Kumar Chaturvedi.
\newblock {\em Network reliability: measures and evaluation}.
\newblock John Wiley \& Sons, 2016.

\bibitem{colbourn1987combinatorics}
Charles~J Colbourn.
\newblock {\em The combinatorics of network reliability}.
\newblock Oxford University Press, Inc., 1987.

\bibitem{Gabow95}
H.N. Gabow.
\newblock A matroid approach to finding edge connectivity and packing
  arborescences.
\newblock {\em Journal of Computer and System Sciences}, 50(2):259--273, 1995.

\bibitem{gomory1961multi}
Ralph~E Gomory and Tien~Chung Hu.
\newblock Multi-terminal network flows.
\newblock {\em Journal of the Society for Industrial and Applied Mathematics},
  9(4):551--570, 1961.

\bibitem{GuoH20}
Heng Guo and Kun He.
\newblock Tight bounds for popping algorithms.
\newblock {\em Random Structures \& Algorithms}, 57(2):371--392, 2020.

\bibitem{HarrisS18}
David~G. Harris and Aravind Srinivasan.
\newblock Improved bounds and algorithms for graph cuts and network
  reliability.
\newblock {\em Random Structures \& Algorithms}, 52(1):74--135, 2018.

\bibitem{Karger99sparsify}
David~R. Karger.
\newblock Random sampling in cut, flow, and network design problems.
\newblock {\em Mathematics of Operations Research}, 24(2):383--413, 1999.

\bibitem{Karger99}
David~R. Karger.
\newblock A randomized fully polynomial time approximation scheme for the
  all-terminal network reliability problem.
\newblock {\em SIAM Journal on Computing}, 29(2):492--514, 1999.

\bibitem{Karger00mincut}
David~R. Karger.
\newblock Minimum cuts in near-linear time.
\newblock {\em J. ACM}, 47(1):46–76, jan 2000.

\bibitem{Karger16}
David~R. Karger.
\newblock A fast and simple unbiased estimator for network (un)reliability.
\newblock In {\em 2016 IEEE 57th Annual Symposium on Foundations of Computer
  Science (FOCS)}, pages 635--644, 2016.

\bibitem{Karger17}
David~R. Karger.
\newblock Faster (and still pretty simple) unbiased estimators for network
  (un)reliability.
\newblock In {\em 2017 IEEE 58th Annual Symposium on Foundations of Computer
  Science (FOCS)}, pages 755--766, 2017.

\bibitem{Karger20}
David~R. Karger.
\newblock A phase transition and a quadratic time unbiased estimator for
  network reliability.
\newblock In {\em Proceedings of the 52nd Annual ACM SIGACT Symposium on Theory
  of Computing}, STOC 2020, page 485–495, New York, NY, USA, 2020.
  Association for Computing Machinery.

\bibitem{KargerKT95}
David~R. Karger, Philip~N. Klein, and Robert~E. Tarjan.
\newblock A randomized linear-time algorithm to find minimum spanning trees.
\newblock {\em J. ACM}, 42(2):321–328, mar 1995.

\bibitem{karger1996new}
David~R Karger and Clifford Stein.
\newblock A new approach to the minimum cut problem.
\newblock {\em Journal of the ACM (JACM)}, 43(4):601--640, 1996.

\bibitem{karp1989monte}
Richard~M Karp, Michael Luby, and Neal Madras.
\newblock Monte-carlo approximation algorithms for enumeration problems.
\newblock {\em Journal of algorithms}, 10(3):429--448, 1989.

\bibitem{Lueker78}
George~S. Lueker.
\newblock A data structure for orthogonal range queries.
\newblock In {\em 19th Annual Symposium on Foundations of Computer Science
  (sfcs 1978)}, pages 28--34, 1978.

\bibitem{Valiant79}
Leslie~G. Valiant.
\newblock The complexity of enumeration and reliability problems.
\newblock {\em SIAM Journal on Computing}, 8(3):410--421, 1979.

\end{thebibliography}

\appendix

\section{Additional Proofs}\label{sec:proofs}\subsection{Proof of Lemma \ref{lem:z-approx-u}}

To prove this lemma, we use the following lemma shown by Karger~\cite{Karger20}:
\begin{lemma}[Theorems 5.1 and 7.1 of \cite{Karger20}]\label{lem:karger-z-approx-u}
There exist thresholds $0<s<b<1$ such that
\begin{enumerate}
    \item When $p<s$, $\frac{x_G(p)}{z_G(p)} \le (p/s)^{\lambda/2}$.
    \item $u_G(b)=\frac 12$, $s^{\lambda} = \Omega(b^\lambda / \log^2 n)$.
    \item $z_{G, \alpha}(b) = n^{-O(1/\log \alpha)}$, where $z_{G, \alpha}(b)$ is the expected number of failed $\alpha\lambda$-weak cuts in $G$ when each edge is removed with probability $b$.
\end{enumerate}
\end{lemma}

Recall that $\xgp$ is the expected number of failed cut pairs, and $\zgp$ is the expected number of failed cuts. A simple inclusion-exclusion on $\ugp$ gives
\[
    \zgp-\xgp \le \ugp \le \zgp.
\]

Let $s$ and $b$ be the thresholds provided by \Cref{lem:karger-z-approx-u}. Set $\theta^\lambda = s^\lambda / \log^2 n$, $\theta<s$. Then when $p<\theta$, 
\[
    \frac{\xgp}{\zgp} \le (p/s)^{\lambda/2} \le (\theta/s)^{\lambda/2} = \frac{1}{\log n},
\]
so the second property holds.

Next we prove the first property. $\theta^\lambda=s^\lambda / \log^2 n = \Omega(b^\lambda / \log^4 n)$.
\begin{align*}
u_G(\theta) &\ge \left(1-\frac{1}{\log n}\right)z_G(\theta) \ge \frac 12 z_{G, \alpha}(\theta)\\
 &= \frac 12 \sum_{C_i:c_i \le \alpha\lambda} \theta^{c_i} \ge \frac 12 \sum_{C_i:c_i \le \alpha\lambda} b^{c_i} \left(\frac{\theta}{b}\right)^{\alpha\lambda}\\
 &\ge z_{G, \alpha}(b) \cdot (\Omega(\log^{-4} n))^\alpha \ge n^{-O(1/\log \alpha)} \cdot e^{-O(\alpha\log \log n)}
\end{align*}

Set $\alpha = O(\log n / (\log \log n)^2)$, then $n^{-O(1/\log \alpha)} = n^{-O(1/\log \log n)}$, and $e^{-O(\alpha\log \log n)} = n^{-O(1/\log \log n)}$. In conclusion $u_G(\theta) = n^{-O(1/\log \log n)}$. So, the first property holds.

\subsection{Proof of Lemma \ref{lem:z-approx-u-reliable}}

We first state a simple fact about the ratio $p^\lambda$, the probability that a specific min cut fails, and $\ugp$, the probability that any cut fails.
\begin{lemma}[Corollary II.2 of \cite{Karger16}]\label{lem:z-trivial-bound}
For a graph $G$ with $n$ vertices and min cut value $\lambda$, 
\[p^\lambda \le \ugp \le n^2p^\lambda.\]
\end{lemma}
By \Cref{lem:z-trivial-bound} and the definition of $b$ in \Cref{lem:karger-z-approx-u}, we have 
\[
    \nf 12 = u_G(b)\le n^2b^\lambda. 
\]    
Therefore, $s^\lambda = \Omega(b^\lambda/\log^2 n ) = \Omega(n^{-2}/\log^2 n)$.
When $p^\lambda = O(n^{-3})$, we have 
\[
    (p/s)^{\lambda/2} =\sqrt{p^\lambda / s^\lambda} \le \sqrt{\frac{O(n^{-3})}{\Omega\left(\frac{n^{-2}}{\log^2 n}\right)}} = O\left(\frac{\log n}{\sqrt{n}}\right). 
\]    
    The lemma follows by plugging this bound into the first property of \Cref{lem:karger-z-approx-u}.

\subsection{Proof of \Cref{lem:tree-packing}}

We use Gabow's algorithm that packs directionless spanning trees in a directed graph $D$. A directionless spanning tree of $D$ is a subgraph of $D$ such that if we replace all its edges with undirected edges, we get a tree spanning all vertices of $D$.

\begin{lemma}[\cite{Gabow95}]\label{lem:gabow-tree-packing}
Given a directed graph $D$ with min-cut value $\lambda$, we can construct in $O(\lambda m \log(n^2/m))$ time a packing of $\lambda$ edge-disjoint directionless spanning trees.
\end{lemma}

Given an undirected graph $G$, we construct a directed graph $D$ by replacing every undirected edge by two directed edges in opposite directions. Then the min-cut value in $D$ is the same as the min-cut value $\lambda$ in $G$. We apply \Cref{lem:gabow-tree-packing} to get $\lambda$ edge-disjoint directionless spanning trees in $D$. After that, remove the directions of all edges in the trees to form a tree packing $\cal T$ in $G$. Removing direction maps two edges in $D$ into one edge in $G$. Therefore, each edge of $G$ is used by at most two trees in $\cal T$.

\subsection{Proof of \Cref{fact:importance-distribution-correct}}

Among the $N^j$ possible sampling sequences, we count the number of sequences that form set $A$. Such a sequence can be generated by first partitioning the $j$ elements into $\alpha$ sets, then choosing a bijection between these sets and elements in $A$. The first step has $S(j, \alpha)$ possibilities, and the second step has $\alpha!$ choices. So, the total number of sequences that yield $A$ is given by $\alpha!\ S(j, \alpha)$. The first part of the lemma now follows since in uniform sampling with replacement, all $N^j$ sequences are equally likely. For the second part, we note that when $j=2$, we have $S(j, \alpha)=1$ for $\alpha\in\{1, 2\}$ and $S(j, \alpha)=0$ otherwise.

\subsection{Proof of \Cref{fact:dfdt}}

Since $t\ge \frac{5}{\log n} = \frac{5\ln 2}{\ln n} > \frac{3}{\ln n}$, we have $2^t-1\ge t\ln 2 > \frac{3\ln 2}{\ln n}$. This implies
$\frac{2^t}{2^t-1} = 1+\frac{1}{2^t-1} < 1+\frac{\ln n}{3\ln 2}$,
and $\frac{2^t \ln 2}{2^t-1} < \ln 2+\frac{1}{3}\ln n$. When $n$ is greater than some constant, we have $\frac{2^t \ln 2}{2^t-1} < 0.375\ln n$. Therefore,
\begin{align*}
    \frac{df}{dt}
    &= (2u)^t \ln (2u) - u^t \ln u
    = u^t \left(2^t\ln 2 - (2^t-1)\ln\frac{1}{u}\right) 
    \le u^t (2^t \ln 2 - (2^t-1) 0.375\ln n)
  < 0.
\end{align*}

\end{document}